\documentclass[USenglish]{lipics-v2019}
\pdfoutput=1 

\nolinenumbers

\theoremstyle{plain}
\newtheorem{mlemma}{Lemma}
\newtheorem{observation}{Observation}

\usepackage{thm-restate}
\usepackage{subcaption}

\newcommand{\ignore}[1]{}

\newcommand{\A}{\mathcal{A}}
\newcommand{\W}{\mathscr{W}}

\newcommand{\ej}{\mathscr{J}}
\newcommand{\eC}{\mathscr{C}}
\newcommand{\si}{\mathrm{I}}
\newcommand{\sI}{\mathscr{I}}

\newcommand{\sX}{\mathscr{X}}

\newcommand{\aj}{\mathfrak{j}}

\newcommand{\eqnote}[1]{\tag{\scriptsize{#1}}}

\makeatletter
\def\moverlay{\mathpalette\mov@rlay}
\def\mov@rlay#1#2{\leavevmode\vtop{%
   \baselineskip\z@skip \lineskiplimit-\maxdimen
   \ialign{\hfil$\m@th#1##$\hfil\cr#2\crcr}}}
\newcommand{\charfusion}[3][\mathord]{
    #1{\ifx#1\mathop\vphantom{#2}\fi
        \mathpalette\mov@rlay{#2\cr#3}
      }
    \ifx#1\mathop\expandafter\displaylimits\fi}
\makeatother

\newcommand{\vol}{\mbox{Vol}}
\renewcommand{\top}{\mbox{Top}}

\newcommand{\dts}{\mbox{Dot}}
\newcommand{\dom}{\mbox{Dom}}

\newcommand{\tI}{t_{I}}

\renewcommand{\P}{\mathscr{P}}
\newcommand{\D}{\mathscr{D}}
\newcommand{\Q}{\mathcal{Q}}
\newcommand{\bad}{\mathscr{B}{\hspace{-0.8mm}\scriptstyle{\mathscr{A}}}{\hspace{-0.8mm}\scriptstyle{\mathscr{D}}}}

\newcommand{\bv}{\mathbf{v}}

\newcommand{\bu}{\mathbf{u}}
\newcommand{\bq}{\mathbf{q}}
\newcommand{\bp}{\mathbf{p}}
\newcommand{\bS}{\mathbf{S}}
\renewcommand{\S}{\mathcal{M}}
\newcommand{\bx}{\mathbf{x}}
\newcommand{\bg}{\mathbf{g}}
\newcommand{\bz}{\mathbf{z}}
\newcommand{\by}{\mathbf{y}}
\newcommand{\bell}{\ensuremath{\boldsymbol\ell}}

\newcommand{\E}{\mathbb{E}}
\newcommand{\G}{\mathscr{G}}

\newcommand{\R}{\mathbb{R}}

\newcommand{\mdef}[1]{{\bf{#1}}}

\DeclareMathAlphabet{\mathpzc}{OT1}{pzc}{m}{it}

\newcommand{\pS}{S^{+}}

\title{A New Lower Bound for Semigroup Orthogonal Range Searching}
\author{Peyman Afshani}{Aarhus University}{peyman@cs.au.dk}{}{supported by DFF (Det Frie Forskningsr\" ad) of Danish Council for Indepndent Reserach under grant ID DFF$-$7014$-$00404.}

\authorrunning{Peyman Afshani}

\Copyright{Peyman Afshani}

\begin{document}

\maketitle
\begin{abstract}
  We report the first improvement in the space-time trade-off of lower bounds
  for the orthogonal range searching problem in the semigroup model, since
  Chazelle's result from 1990. 
  This is one of the very fundamental problems in range searching with a long history.
  Previously, Andrew Yao's influential result had shown that the problem is
  already non-trivial in one dimension~\cite{Yao-1Dlb}:
  using $m$ units of space, the query time $Q(n)$ must 
  be $\Omega( \alpha(m,n) + \frac{n}{m-n+1})$ where $\alpha(\cdot,\cdot)$
  is the inverse Ackermann's function, a very slowly growing function. 
  In $d$ dimensions, Bernard Chazelle~\cite{Chazelle.LB.II} proved that 
  the query time must be $Q(n) = \Omega( (\log_\beta n)^{d-1})$ where
  $\beta = 2m/n$.
  Chazelle's lower bound is known to be tight for when space consumption is ``high''
  i.e., $m = \Omega(n \log^{d+\varepsilon}n)$. 

  We have two main results.
  The first is a lower bound that shows Chazelle's lower bound was not tight for ``low space'':
  we prove that we must have $m Q(n) = \Omega(n (\log n \log\log n)^{d-1})$.
  Our lower bound does not close the gap to the existing data structures, however, our second
  result is that our analysis is tight. 
  Thus, we believe the gap is in fact natural since lower bounds are proven for
  idempotent semigroups while the data structures are built for general semigroups
  and thus they cannot assume (and use) the properties of an idempotent semigroup.
  As a result, we believe to close the gap one must study lower bounds for non-idempotent semigroups
  or building data structures for idempotent semigroups. 
  We develope significantly new ideas for both of our results 
  that could be useful in pursuing either of these directions.
\end{abstract}

\section{Introduction}
Orthogonal range searching in the semigroup model is one of the most fundamental data structure problems in computational geometry. 
In the problem, we are given an input set of points to store in a data structure where each point is associated with a weight
from a semigroup $\G$ and the goal is to compute the (semigroup) sum of all the weights inside an axis-aligned box given at the query time.
Disallowing the ``inverse'' operation in $\G$ makes the data structure very versatile as it is then applicable to a wide range
of situations (from computing weighted sum to computing the maximum or minimum inside the query). 
In fact, the semigroup variant is the primary way the family of range searching problems are introduced,
(see the survey~\cite{Agarwal.survey16}). 

Here, we  focus only on static data structures.
We use the convention that 
$Q(n)$, the query time, refers to the worst-case number of semigroup additions required to produce the query answer
$S(n)$, space, refers to the number of semigroup sums stored by the data structure.
By storage, denoted by $\pS(n)$, we mean space but not counting the space used by the input, i.e.,
$S(n) = n + \pS(n)$.
So we can talk about data structures with sublinear space, 
e.g., with 0 storage the data structure has to use the input weights only, leading to
the worst-case query time of $n$.

\subsection{The Previous Results}
Orthogonal range searching is a fundamental problem with a very long history. 
The problem we study is also very interesting from a lower bound point of view where the
goal is to understand the fundamental barriers and limitations of performing
basic data structure operations.
Such a lower bound approach 
was initiated by Fredman in early 80s and in a series of very influential papers (e.g., see~\cite{Fredman.LB.80,Fredman.LB.semigroup,Fredman.LB81}).
Among his significant results, was the lower bound~\cite{Fredman.LB.semigroup,Fredman.LB81} that showed
a sequence of $n$ insertions, deletions, and queries requires $\Omega(n \log n)$ time to run. 

Arguably, the most surprising result of these early efforts was given by Andrew Yao who in 1982
showed that even in one dimension, the static case of the problem contains a very non-trivial, albeit small, barrier.
In one dimension, the problem essentially boils down to adding numbers: store an input array $A$ of $n$ numbers in a data structure 
s.t., we can add up the numbers from $A[i]$ to $A[j]$ for $i$ and $j$ given at the query time. 
The only restriction is that we should use only additions and not subtractions (otherwise, the problem is easily solved using
prefix sums). 
Yao's significant result was that answering queries requires
$\Omega(\alpha(S(n),n) + n/\pS(n))$ additions, where $\alpha(\cdot,\cdot)$ is the inverse Ackermann function. 
This bound implies that if one insists on using $O(n)$ storage, the query bound cannot be reduced to constant,
but even using a miniscule amount of extra storage (e.g., a $\log^*\log^* n$ factor extra storage) can reduce the query bound 
to constant. 
Furthermore, using a bit less than $n$ storage, e.g., by a $\log^*\log^* n$ factor, will once again yield a more natural (and optimal) bound of $n/\pS(n)$.
Despite its strangeness, 
it turns out there are data structures that can match the exact lower bound (see also~\cite{Alon.opt1D}).
After Tarjan's famous result on the union-find problem~\cite{Tarjan-UF-79}, 
this was the second independent appearance of the inverse Ackermann function in the history of algorithms and data structures.

Despite the previous attempts, the problem is still open even in two dimensions.
At the moment, using range trees~\cite{Bentley.79,Bentley.80} on the 1D structures 
is the only way to get two or higher dimensional results. 
In 2D for instance, we can have $\pS(n) = O(n/\log n)$ with query bound $Q(n) = O(\log^3 n)$, or 
$\pS(n) = O(n)$ with query bound $Q(n) = O(\log^2 n)$, or 
$\pS(n) = O(n\log n)$ with query bound $O(\alpha(cn,n)\log n)$, for any constant $c$.
In general and in $d$ dimensions, we can build a  structure  with $\pS(n) = O(n\log^{d-1}n)$ units of storage and with 
$Q(n)= O(\alpha(c n,n)  \log^{d-1}n)$ query bound, for any constant $c$.
We can reduce the space complexity by any factor $t$ by increasing the query bound by another factor $t$.
Also, strangely, if $t$ is asymptotically larger than $\alpha(n,n)$, then the inverse Ackermann term in the query 
bound disappears. 
Nonetheless, a surprising result of Chazelle~\cite{Chazelle.LB.II} shows that the reverse is not true:
the query bound must obey $Q(n) = \Omega( (\log_{S(n)/n}n)^{d-1})$ which implies using polylogarithmic
extra storage only reduces the query bound by a $(\log\log n)^{d-1}$ factor. 
Once again, using range tree with large fan out, one can build a data structure
that uses $O(n\log^{2d-2+\varepsilon}n)$ storage, for any positive constant $\varepsilon$, and achieves the
query bound of $O( (\log_{\log n}n)^{d-1})$.
This, however leaves a very natural and important open problem:
{\it Is Chazelle's lower bound  the only barrier? Is it possible to achieve
    $O(n)$ space and $O(\log^{d-1}n)$ query time?}

\subparagraph{Idempotence and random point sets.}
A semigroup is idempotent if for every $x \in \G$, we have $x+x = x$.
All the previous lower bounds are in fact valid  for idempotent semigroups.
Furthermore, Chazelle's lower bound uses a uniform (or randomly placed) set of points which shows 
the lower bound does not require pathological or fragile input constructions.
Furthermore, his lower bound also holds for dominance ranges, i.e., $d$-dimensional boxes in the form
of $(-\infty, a_1] \times \dots \times (-\infty, a_d]$.
These little perks result in  a very satisfying statement: 
problem is still difficult even when $\G$ is ``nice'' (idempotent), and when
the point set is ``nice'' (uniformly placed) and when the queries are simple (``dominance queries''). 

\subsection{Our Results}
We show that for any data structure that uses $\pS(n)$ storage and has query bound of $Q(n)$,
we must have $\pS(n) \cdot Q(n) = \Omega(n (\log n \log\log n)^{d-1})$.
This is the first improvement to the storage-time trade-off curve for the problem since Chazelle's result in 1990.
It also shows that Chazelle's lower bound is not the only barrier.
Observe that our lower bound is strong at a different corner
of parameter space compared to Chazelle's:
ours is strongest when storage is small whereas Chazelle's is strongest when the storage is large. 
Furthermore, we also keep most of the desirable properties of Chazelle's lower bound:
our lower bound also holds for  idempotent semigroups and uniformly placed point sets.
However, we have to consider more complicated queries than just dominance queries which ties to our second main result. 
We show that our analysis is tight: 
given a ``uniformly placed'' point set and an idempotent semigroup $\G$, 
we can construct a data structure that uses $O(n)$ storage
and has the query bound of $O( (\log n \log\log n)^{d-1})$.
As a corollary, we provide an almost complete understanding of orthogonal range searching queries with respect to a uniformly
placed point set in an idempotent semigroup. 

\subparagraph{Challenges.}
Our results and specially our lower bound require significantly new ideas.
To surpass Chazelle's lower bound, we need to go beyond dominance queries which 
requires
wrestling with complications that ideas such as range trees can introduce.
Furthermore, in our case, the data structure can actually improve the query time 
by a factor $f$ by spending a factor $f$ extra space.
This means, we are extremely sensitive to how the data structure can ``use'' its space.
As a result, we need to capture the limits of how intelligently the
data structure can spend its budge of ``space'' throughout various subproblems.

\subparagraph{Implications.}
It is natural to conjecture that the uniformly randomly placed point set should be the most difficult point set 
for orthogonal queries. 
Because of this, we conjecture that our lower bounds are almost tight. 
This opens up a few very interesting open problems. 
See Section~\ref{sec:conc}.

\section{Preliminaries}

\subparagraph{The Model of Computation.}
Let $P$ be an input set of $n$ points with weights  from a semigroup $\G$. 
Our model of computation is the same as the one used by the previous lower bounds, e.g.,~\cite{Chazelle.LB.II}.
There has been quite some work dedicated to building a proper model for lower bounds
in the semigroup model. 
We will not delve into those details and we only mention the final consequences of the efforts.
The data structure stores a number of sums where 
each sum $s$ is the sum of the weights of a subset $s_P \subset P$. 
With a slight abuse of the notation, we will use $s$ to refer both to the sum as well as to the subset $s_P$. 
The number of stored sums is the space complexity of the data structure. 
If a sum contains only one point, then we call it a \mdef{singleton}
and we use $\pS(n)$ to denote the storage occupied by sums that are not singletons. 
Now, consider a query range $r$ containing a subset $r_P = r \bigcap P$. 
The query algorithm must 
find $k$ stored subsets $s_1, \dots, s_k$ such that
$r_P = \cup_{i=1}^k s_i$.
For a given query $r$, the smallest such integer $k$ is the query bound of the query.
The query bound of the data structure is the worst-case query bound of any query.
Observe that the data structure does not disallow covering any point more than once and in fact,
for idempotent semigroups this poses no problem.
All the known lower bounds work in this way, i.e., they 
allow covering a point inside the query multiple times. 
However, if the semigroup is not idempotent, then covering a point more than once could lead to
incorrect results. 
Since data structures work for general semigroups, they 
ensure that 
$s_1, \dots, s_k$ are disjoint. 

\subparagraph{Definitions and Notations.}
A $d$-dimensional dominance query is determined by one point $(x_1, \dots, x_d)$ and it is
defined as $(-\infty,x_1] \times \dots \times (-\infty,x_d]$.

\ignore{
In the context of lower bounds for idempotent semigroups,
we can make the following simplifying assumptions.
Consider a stored sum $s$ and
let $x_\ell, x_r$ and $y_t$ be the smallest $X$-coordinate,
the largest $X$-coordinate, and
the largest $Y$-coordinate of any point in $s$, respectively.
Now, we can update $s$ by adding any point inside the range $[x_\ell, x_r] \times (-\infty, y_t]$ to the sum $s$. 
Clearly, if the old sum $s$ was used to cover the points inside a query region $r$, then the range
$[x_\ell, x_r] \times (-\infty, y_t]$ had to be inside the region $r$ which implies the updated sum can
also be used (and in fact it can potentially cover more points).
We call the point $(x_r, y_t)$ the \mdef{dot} of $s$ and denote it with $\dts(s)$.
We call the line segment $[x_\ell, x_r]\times y_t$ the \mdef{(top) marker} of $s$ and denote it with $t(s)$.
Finally, for a rectangle $r$, we define the $X$-range of $r$ as the interval obtained by its projection 
onto the $X$-axis. 
$Y$-range is defined similarly.
}

\begin{definition}\label{def:wd}
We call a set $P \subset \R^d$  \mdef{well-distributed}  
  if the following properties hold:
  (i) $P$ is contained in the $d$-dimensional unit cube.
  (ii) The volume of any rectangle that contains $k\ge 2$ points of $P$ is  at least $\varepsilon_dk/|P|$
  for some constant $\varepsilon_d$ that only depends on the dimension.
  (iii) Any rectangle that has volume $v$, contains at most $\lceil v |P| / \varepsilon_d \rceil$ points of $P$.
\end{definition}

\begin{mlemma}\label{lem:well}\cite{aal12,Chazelle.LB.II,AD.frechet17}
  For any constant $d$ and any given value of $n$, there exists a
  well-distributed point set in $\R^d$ containing $\Theta(n)$ points.
\end{mlemma}

\section{The Lower Bound}\label{sec:lb}
This section is devoted to the proof of our main theorem which is the following.

\begin{theorem}
  If $P$ is a  well-distributed point set of $n$ points in $\R^d$, any data structure that uses $\pS(n)$ storage,
  and answers $(2d-1)$-sided queries in
  $Q(n)$ query bound requires that 
  $\pS(n)\cdot Q(n) = \Omega(n (\log n \log\log n)^{d-1})$.
\end{theorem}

Let $\Q$ be the unit cube in $\R^d$. 
Throughout this section, the input point set is a set $P$ of $n$  well-distributed points in $\Q$.
Let $\D$ be a data structure that answers semigroup orthogonal range searching queries on $P$.

\subsection{Definitions and Set up}
We consider queries that have two boundaries in dimensions 1 to $d-1$ but only have
an upper bound in dimension $d$. 
For simplicity, we rename the axes such that the $d$-th axis
is denoted by $Y$ and the first $d-1$ axes are denoted by $X_1, \dots, X_{d-1}$. 
Thus, each query is in the form of $[x'_1,x_1]\times \dots [x'_{d-1},x_{d-1}]\times (-\infty,y]$.
The point $(x_1, \dots, x_{d-1},y)$ is defined as the \mdef{dot of $q$} and is denoted
by $\dts(q)$.
For every $1 \le i \le d-1$, the line segment that connects $\dts(q)$ to the point 
$(x_1, \dots, x_{i-1}, x'_{i}, x_{i+1}, \dots, x_{d-1},y)$ is called the $i$-th
marker of $q$ and it is denoted by $t_i(s)$.

\subparagraph{The tree $T_i$.}
For each dimension $i=1, \dots, d-1$, we define a balanced binary tree $T_i$ 
of height $h = \log n$ as follows.
Informally, we cut $\Q$ into $2^h$ congruent boxes with
hyperplanes perpendicular to axis $X_i$ which form the leaves of $T_i$.
To be more specific, every node in $T_i$ is assigned a box $r(v) \subset \Q$. 
The root of $T_i$ is assumed to have depth $0$ and it is assigned $\Q$. 
For every node $v$, we divide $r(v)$ into two congruent ``left'' and ``right'' boxes with
a hyperplane $\ell(v)$, perpendicular to $X_i$ axis. 
The left box is assigned to left child of $v$ and similarly the right box is assigned to the right child of $v$. 
We do not do this if $r(v)$ has volume less than $1/n$; these nodes become the leaves of $T$.
Observe that all trees $T_i$, $1 \le i \le d-1$ have the same height $h$. 
The volume of $r(v)$ for a node $v$ at depth $j$ is  $2^{-j}$.

\subparagraph{Embedding the problem in $\R^{2d-1}$.}
The next idea is to embed our problem in $\R^{2d-1}$.
Consistent with the previous notation, the first $d$ axes are $X_1, \dots, X_{d-1}$ and $Y$.
We label the next axis $Z_1, \dots, Z_{d-1}$.
We now represent $T_i$ geometrically as follows. 
Consider $h$ the height of $T_i$.
For each $i$, $1 \le i \le d-1$, 
we now define a \mdef{representative diagram $\Gamma_i$} which is a axis-aligned
decomposition of the unit (planar) square $Q_i$ in a coordinate system where the horizontal axis is $X_i$
and the vertical axis is $Z_i$.
As the first step of the decomposition, cut $Q_i$ into $h$ equal-sized sub-rectangles using $h-1$ horizontal lines.
Next, we will further divide each sub-rectangle into small regions and we will assign every
node $v$ of $T_i$ to one of these regions.
This is done as follows. 
The root $v$ of $T_i$ is assigned the topmost sub-rectangle as its region, $\gamma(v)$.
Assume $v$ is assigned a rectangle $\gamma(v)$ as its region.
We create a vertical cut starting from the middle point of the lower boundary of $\gamma(v)$ 
all the way down to the bottom of the rectangle $Q_i$. 
The children of $v$ are assigned to the two rectangles that lie immediately below $\gamma(v)$.
See Figure~\ref{fig:int}.

\begin{figure}[h]
  \centering
  \includegraphics[scale=0.65]{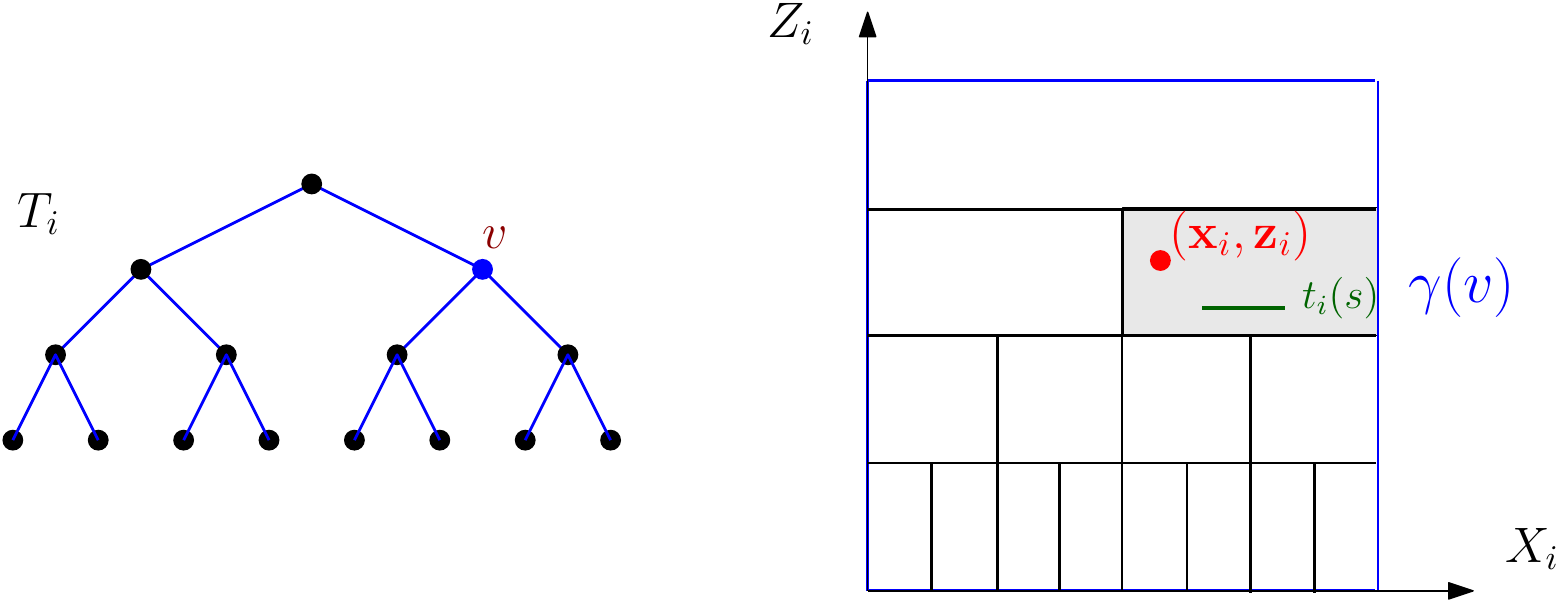}
  \caption{A tree and its representative diagram. The region of a node $v$ is highlighted in grey.}
  \label{fig:int}
\end{figure}

\subparagraph{Placing the Sums.}
Consider a semigroup sum $s$ stored by the data structure $\D$.
Our lower bound will also apply to semigroups that are idempotent which means without loss
of generality, we can assume that our semigroup is idempotent. 
As a result, we can assume that each semigroup sum $s$ stored by the data structure has the same
shape as the query.
Let $b(s)$ be the smallest box that is unbounded from below (along the $Y$ axis) that contains
all the points of $s$.
If $s$ does not include a point, $p$, inside $b(s)$, we can just  $p$ to $s$.
Any query that can use $s$ must contain the box $b(s)$ which means adding $p$ to $s$
can only improve things. 
Each sum $s$ is placed in one node of $T_i$ for every $1 \le i \le d-1$.
The details of this placement are as follows. 

A node $v_i$ in $T_i$ stores any sum $s$ such that the $i$-th marker of $s$, $t_i(s)$, intersects 
$\ell(v_i)$  with $v$ being the {\em highest} node with this property. 
Geometrically, this is equivalent to the following: we place $s$ at a node $v$ if $\gamma(v)$ is the
{\em lowest} region that fully contains the segment $t_i(s)$ (or to be precise, the projection  of
$t_i(s)$ onto the $Z_iX_i$ plane).
For example, in Figure~\ref{fig:int}(right), the sum $s$ is placed at $v$ in $T_i$ since
$t_i(s)$, the green line segment, is completely inside $\gamma(v)$ with $v$ being the lowest node of this property.
Remember that $s$ is placed at some node in each tree $T_i$, $1 \le i \le d-1$ (i.e., it is placed $d-1$ times in total).

\ignore{
\begin{observation}
  Our placement gives us three important properties:
  (i) For every ancestor $u$ of $v$, $s$ lies completely to one side of the hyperplane $\ell(u)$.
  (ii) $s$ contains one point from each side of $\ell_{v_i}$.
  (iii) For every descendant $w$ of $v$, $s$ contains the weight of a point that is outside $r(w)$. 
\end{observation}
}

\subparagraph{Notations and difficult queries.}
We will adopt the convention that random variables are denoted with bold math font. 
The difficult query is a $2d-1$ sided query chosen randomly as follows. 
The query is defined as $[\bx'_1, \bx_1]\times \dots \times [\bx'_{d-1},\bx_{d-1}] \times (-\infty,\by]$
where $\bx'_i, \bx_i$ and $\by$ are also random variables (to be described).
$\by$ is chosen uniformly in $[0,1]$.
To choose the remaining coordinates, we do the following. 
We place a random point $(\bx_i,\bz_i)$ uniformly inside the representative
plane $\Gamma_i$ (i.e., choose $\bx_i$ and $\bz_i$ uniformly in $[0,1]$).
Let $\bv_i$ be the random variable denoting the node in $T_i$ 
s.t., region $\gamma(\bv_i)$ contains the point $(\bx_i,\bz_i)$.
$\bx'_i$ is the $X_i$-coordinate of the left boundary of $\gamma(\bv_i)$.
Let $\bell_i$ be the depth of $\bv_i$ in $T_i$.
We denote the point $(\bx_1, \dots, \bx_{d-1},\by)$ by $\bq$ and denote the $2d-1$ sided query
by $\dom_{\bv_1, \ldots, \bv_{d-1}}(\bq)$.
See Figure~\ref{fig:int}(right).
Note that a query $\dom_{v_1, \ldots, v_{d-1}}(q)$ is equivalent to a dominance query defined by point $q$
in $r(v_1)\cap \dots \cap r(v_{d-1})$. 
To simplify the presentation and to stop redefining these concepts, we will reserve the notations introduced in this
paragraph to only represent the concepts introduced here.

\begin{observation}\label{ob:highsubtree}
    A necessary condition for being able to use a sum $s$ to answer
 $\dom_{v_1, \dots, v_{d-1}}(q)$  is that  $s$ is stored at the subtree of $v_i$,
for every $1 \le i \le d-1$.
\end{observation}
\begin{proof}
 Due to how we have placed the sums, 
 the sums stored at the ancestors of $v_i$ contain at least one point that lies outside $r(v_i)$
 and since $\dom(q)$ is entirely contained inside $r(v_i)$ those sums cannot be used to answer the query. 
\end{proof}

\subparagraph{Subproblems.}\label{par:sub}
Consider a query $\dom(q)=\dom_{v_1, \dots, v_{d-1}}(q)$.
We now define subproblems of $\dom(q)$.
A subproblem is represented by an array of $d-1$ integral indices $\aj = (j_1, \dots, j_{d-1})$ and it is denoted
as $\aj$-subproblem. 
The state of $\aj$-subproblem of a query $\dom_{v_1, \ldots, v_{d-1}}(q)$ could either be {\em undefined},
or it could refer to covering a particular subset of points inside the query.
In particular, given $\dom(q)$, a $\aj$-subproblem is undefined if  for some $1 \le i \le d-1$, there is no node 
$u_i \in T_i$ with the following properties:
$u_i$ has depth $\ell_i + j_i$,  $u_i$ has a right sibling
$u'_i$ with $r(u'_i)$ containing the query point $q$. 
See Figure~\ref{fig:highsub2}.
However, if such nodes $u_i$ exist for all $1 \le i \le d-1$, then the
$\aj$-subproblem of $\dom(q)$ is \mdef{well-defined} and it refers to the problem of covering all the points inside the
region $\dom(q) \cap r(u_1) \cap \dots \cap r(u_{d-1})$; observe that this is equivalent to covering
all the points inside the region $r(u_1) \cap \dots \cap r(u_{d-1})$ that have $Y$-coordinate at most $y$.
Further observe that for $u_i$ to exist in $T_i$, it needs to pass two \mdef{checks}:
\mdef{(check I)} $\ell_i + j_i \le h$ as otherwise, there are no nodes with depth $\ell_i + j_i$ and
\mdef{(check II)} a node $u_i$ at depth $\ell_i + j_i$ has a right sibling $u'_i$ with $r(u'_i)$ containing $q$.
The nodes $u_1, \dots, u_{d-1}$ are called the \mdef{defining nodes} of the $\aj$-subproblem.
Thus, the random variable $\bv_i$ defines the random variable $\bu_i$ where $\bu_i$ could be
either undefined or it could be a node in $T_i$.
Clearly, the distribution of $\bu_i$ is independent of the distributions of
$\bu_j$ and $\bv_j$ for $i \not = j$  as $\bu_i$ only depends on $\bv_i$.

\begin{restatable}{observation}{obhighsub}\label{ob:highsubproblem}
    Consider a well-defined $\aj = (j_1, \dots, j_{d-1})$ subproblem of a query $\dom_{v_1, \dots, v_{d-1}}(q)$
    and its defining nodes $u_1, \dots, u_{d-1}$.
    To solve the $\aj$-subproblem (i.e., to cover the points inside the subproblem),
    the data structure can use a sum $s$ only if for every $1 \le i \le d-1$,
    we either have case (i) where $s$ is stored at ancestors of $u_i$ but not the ancestors of $v_i$
    or case (ii) where $s$ is stored at the subtree of $u_i$. 
    If a sum $s$ violates one of these two conditions for some $i$, then it cannot
    be used to answer the $\aj$-subproblem.
    See Figure~\ref{fig:highsub2}.
\end{restatable}
\begin{proof}
    First we use Observation~\ref{ob:highsubtree}.
    $s$ must be stored at the subtree of $v_i$. 
    Let $w$ be the node that stores $s$.
    If $w$ is in the subtree of $u_i$, then we are done.
    Otherwise, let $v'$ be the least common ancestor of $u_i$
    and $w$. 
    If $v' = w$ then we are done again but otherwise,
    $u$ belongs to the subtree of one child of $v'$ while $w$ belongs to the
    subtree of the other child of $v'$. 
    By our placement rules, this implies that $s$ is entirely outside $r(u_i)$ and thus it cannot
    be used to answer the $\aj$-subproblem. 
\end{proof}

\begin{figure}[t]
  \centering
  \includegraphics[scale=0.5]{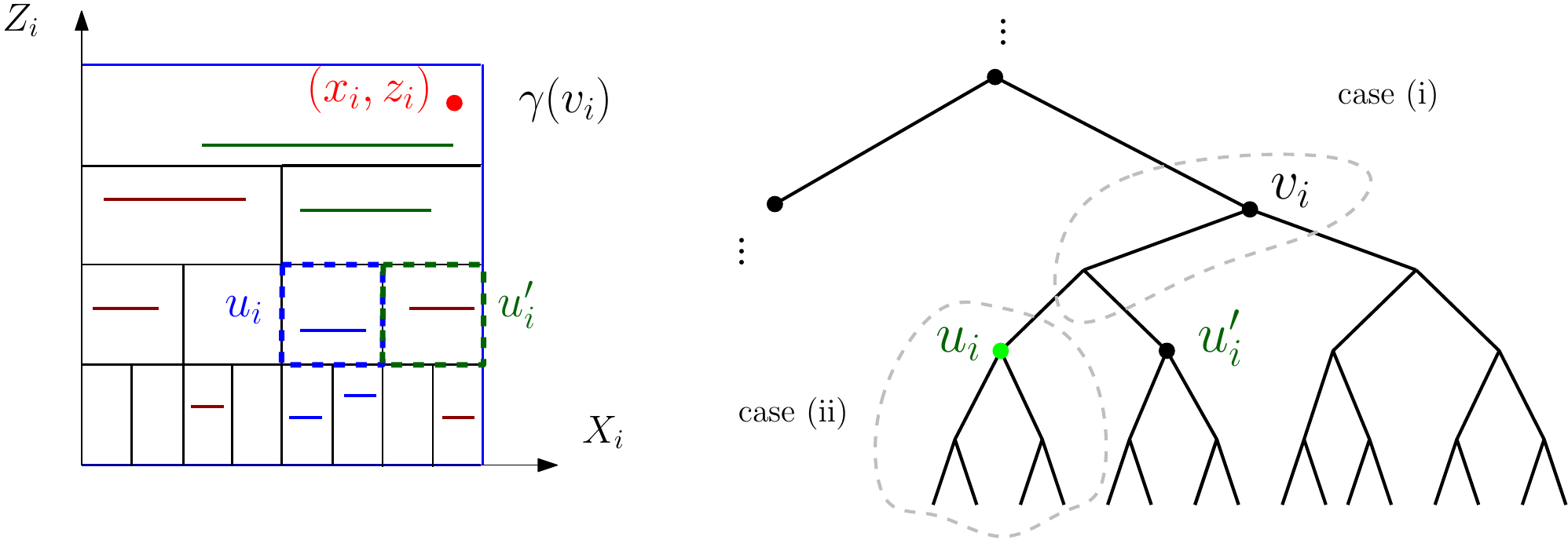}
  \caption{$u_i$ is a defining node. The blue line segments correspond to $t_i(s)$ of a sum $s$ that is placed in the subtree
  of $u_i$. The green ones correspond to those placed at ancestors of $u_i$ but not at ancestors of $v_i$.
  The red ones correspond to sums that cannot be used to answer the subproblem.}
  \label{fig:highsub2}
\end{figure}

\subsection{The Main Lemma}
\begin{figure}[h]
    \centering
    \includegraphics{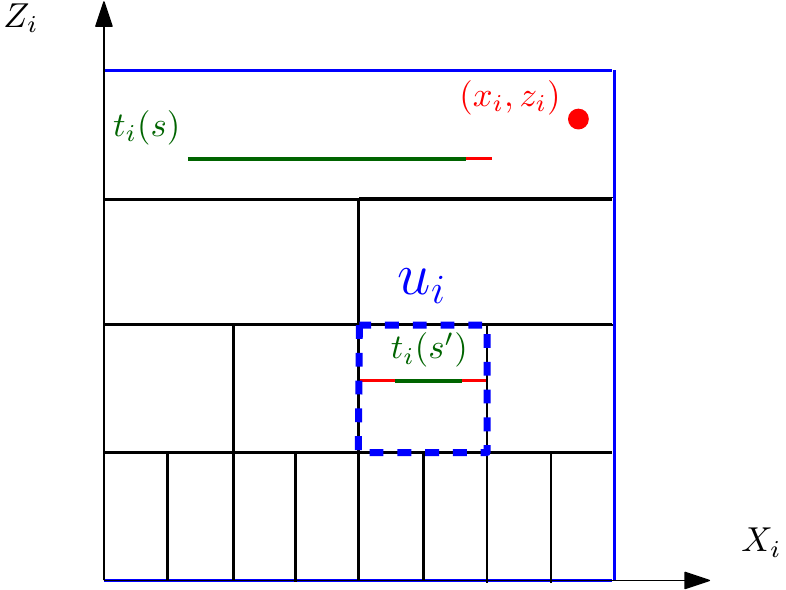}
    \caption{The extensions of two sums that can be used to answer a subproblem.}
    \label{fig:cover}
\end{figure}

In this subsection, we prove a main lemma which is the heart of our lower bound proof. 
To describe this lemma, we first need the following notations.
Consider a well-defined $\aj$-subproblem of a query $\dom_{v_1, \ldots, v_{d-1}}(q)$ where $j_i \le \frac{h}2$ for 
$1 \le i \le d-1$. 
As discussed, this subproblem corresponds to covering all the points in the region
$r(u_1) \cap \cdots \cap r(u_{d-1})$ whose $Y$-coordinate is below $y$, the $Y$-coordinate of point $q$;
thus, the $\aj$-subproblem of the query can be represented as the problem of covering all the points inside the box
$[a_1, b_1] \times \cdots \times [a_{d-1}, b_{d-1}] \times (-\infty, y]$ where $a_i$ and $b_i$ correspond to the
left and the right boundaries of the slab $r(u_i)$.
Let $0 < \lambda$ be a parameter. 
Consider the region 
$[a_1, b_1] \times \cdots \times [a_{d-1}, b_{d-1}] \times [y-\beta, y]$  in which 
$\beta$ is chosen such that the region contains $\lambda$ points;
as our pointset is well-distributed, this implies that the volume of the region is
$\Theta(\lambda/n)$.
We call this region \mdef{the $\lambda$-top box}.
The \mdef{$\lambda$-top}, denoted  by $\top(\aj,\lambda)$,
is then the problem of covering all the points inside the $\lambda$-top box of the $\aj$-subproblem. 
With a slight abuse of the notation, we will use $\top(\aj,\lambda)$ to refer also to the set of points
inside the $\lambda$-top box.
If there are not enough points in the $\lambda$-top box, the $\lambda$-top 
is undefined, otherwise, it is well-defined. 
\newcommand{\tp}{\mbox{Top}}
These of course also depend on the query but we will not write the dependency on the query as it will clutter the notation.
Furthermore,  observe that when the query is random, then $\top(\aj,\lambda)$ becomes a random variable which is either
undefined or it is some subset of points. 

\subparagraph{Extensions of sums.}
Due to technical issues, we slightly extend the number of points each sum covers. 
Consider a sum $s$ stored at a subtree of $v_i$ such that $s$ can be used to answer the $\aj$-subproblem. 
By Observation~\ref{ob:highsubtree}, $s$ is either placed at the subtree of $u_i$ or on the path connecting
$v_i$ to $u_i$. 
We extend the $X_i$ range of the sum $s$ (i.e., the projection of $s$ on the $X_i$)  to include the left and the right boundary of the
node $u_i$ along the $X_i$-dimension. 
We do this for all $d-1$ first dimensions to obtain an extension $e(s)$ of sum $s$.
We allow the data structure to cover any point in $e(s)$ using $s$. 

\begin{restatable}{mlemma}{mainlb}{\normalfont [The Main Lemma]}
    \label{lem:mainlb}
  Consider a $\aj = (j_1, \dots, j_{d-1})$ subproblem of a random query $\dom_{\bv_1, \dots, \bv_{d-1}}(\bq)$, for
  $1 \le j_i \le h/2$.
  Let $\lambda = \frac{\delta h^{d-1}}{j_1 j_2 \dots j_{d-1}}\cdot \frac{n}{\pS(\A)}$ where $\delta$ is a small enough constant
  and $\pS(\A)$ is the storage of the data structure.
  Let $\bS_\aj$ be the set of sums $s$ such that (i) $s$ is contained inside the query $\dom_{\bv_1, \dots, \bv_{d-1}}(\bq)$,
  and (ii) $e(s)$ covers at least $C$ points from $\top(\aj,\lambda)$, meaning, $|e(s) \cap \top(\aj,\lambda)| \ge C$
  where $C$ is  a large enough constant.

  With $\Omega(1)$ probability,  the $\aj$-subproblem and the  $\top(\aj,\lambda)$ are well-defined.
  Furthermore conditioned on both of these being well-defined, 
  with probability $1-O(\sqrt{\delta/\varepsilon_d})$, the nodes $\bv_1, \cdots, \bv_{d-1}$ will
  be sampled as nodes $v_1, \cdots, v_{d-1}$, s.t., the following holds:
  $\E[\sum_{s \in \bS_\aj} |e(s) \cap \top(\aj,\lambda)|] < \frac{|\top(\aj,\lambda)|}{C'}$ where the expectation
  is over the random choices of $\by$ and $C'$ is another large constant.
\end{restatable}
 
Let us give some intuition on what this lemma says and why it is critical for our lower bound.
For simplicity assume $\pS(\A) = n$ and assume we sample $\bv_1, \cdots, \bv_{d-1}$ as the first step, and and then 
sample $\by$ as the last step. 
The above lemma  implies that if we focus on one particular subproblem, the sums in the data structure cannot
cover too many points; to see this consider the following. 
The lemma first says that after the first step, with positive constant probability, $\aj$-subproblem and $\top(\aj,\lambda)$ are well-defined.
Furthermore, here is a very high chance that our random choices will ``lock us'' in a ``doomed'' state, 
after sampling $v_1, \cdots, v_{d-1}$.
Then, when considering the random choices of $\by$, 
sums that cover at least $C$ points in total cover a very small fraction of the points.
As a result, we will need $\Omega(\lambda/C) = \Omega(\frac{\delta\log^{d-1}n}{C j_1 j_2 \dots j_{d-1}})$ sums to cover the points inside
the $\lambda$-top of the subproblem. 
Summing these values over all possible subproblems, 
$j_i, 1 \le j_i \le h/2$, $1 \le i \le d-1$ will create a lot of Harmonic sums of the type
$\sum_{x=1}^{h/2}x = O(\log \log n)$ which will eventually lead to our lower bound. 
In particular,  we will have $\sum_{j_i, 1 \le j_i \le h/2}\Omega(\frac{h^{d-1}}{j_1 j_2 \dots j_{d-1}}) = (\log n \log\log n)^{d-1}.$ 
There is however, one very big technical issue that we will deal with later: a sum can cover very few points from each
subproblem but from very many subproblems!
Without solving this technical issue, we only get the bound $\max_{j_i, 1 \le j_i \le h/2}\Omega(\frac{h^{d-1}}{j_1 j_2 \dots j_{d-1}}) = (\log n)^{d-1}$ 
which offers no improvements over Chazelle's lower bound. 
Thus, while solving this technical issue is important, nonetheless, it is clear that the lemma
we will prove in this section is also very critical. 

As this subsection is devoted to the proof of the above lemma, we will assume that we are considering
a fixed $\aj$-subproblem and thus the indices $j_1, \ldots, j_{d-1}$ are fixed. 
 
\subsubsection{Notation and Setup} 
By Observation~\ref{ob:highsubproblem}, only a particular set of sums can be used to answer the 
$\aj$-subproblem of a query. 
Consider a sum $s$ that can be used to answer the subproblem of some query.
By the observation, we must have that $s$ must either satisfy case (i) or case (ii) for every
tree $T_i$, $1\le i \le d-1$. 
Over all indices $i$, $1\le i \le d-1$, they describe $2^{d-1} = O(1)$ different cases. 
This means that we can partition $\bS_\aj$ into $2^{d-1}$ different \mdef{equivalent
classes} s.t., for any two sums $s_1$ and $s_2$ in an equivalent class,  either they both satisfy case (i) or 
they both satisfy case (ii) in Observation~\ref{ob:highsubproblem} and for any dimension
$i$. 
Since $2^{d-1}$ is a constant, it suffices 
to show that our lemma holds when only considering sums of particular equivalent class.
In particular, let $S'_\aj$ be the subset of eligible sums that all belong to one equivalent class. 
Now, it suffices to show that 
$\E[\sum_{s \in S'_\aj} |e(s) \cap \top(\aj,\lambda)|] < \frac{|\top(\aj,\lambda)|}{2^d C'}$.
since summing these over all $2^{d-1}$ equivalent classes will yield the lemma. 
Furthermore, w.l.o.g and by renaming the $X$-axes, we can assume that there exists 
a fixed value $t$, $0 \le t \le d-1$, such that for every sum $s \in S'_\aj$,
for dimensions $1 \le i \le t$, $s$ satisfies case (i) in $T_i$ and for
$t < i \le d-1$, $s$ is within case (ii). 
Note that if $t=0$, then it implies that we have no instances of case (i) and for
$t=d-1$ we have no instances of case (ii).

\subparagraph{The probability distribution of subproblems.}
To proceed, we need to understand the distribution of the subproblems.
This is done by the following observation. 
\begin{restatable}{observation}{obdistu}\label{ob:distu}
    Consider a $\aj = (j_1, \dots, j_{d-1})$ subproblem of a random query $\dom_{\bv_1, \dots, \bv_{d-1}}(\bq)$
    defined by random variables $\bu_1, \dots, \bu_{d-1}$.
    We can make the following observations.
    (i) the distribution of the random variable $j_i + \bell_i$ is uniform among the
    integers $j_i + 1, \dots, h+j_i$.
    (ii) With probability  $j_i/h$, $\bu_i$ will be undefined because it fails (Check I).
    (iii) If (Check I) does not fail for $\bu_i$, there is exactly 0.5 probability that $\bu_i$
    is undefined.
    (iv) For  a fixed $j_i$,  the probability distribution, $\mu_i$, of $\bu_i$ is as follows:
    with probability $1-\frac{1-j_i/h}2$, $\bu_i$ is undefined. 
    Otherwise, $\bu_i$ is a node in $T_i$ sampled in the following way:
    sample a random integer (depth) $\bell'$ uniformly among integers in $j_i+1, \dots, h$
    and select a random node uniformly among all the nodes at depth $\bell'$ that have a right sibling.
\end{restatable}
\begin{proof}
  (i) follows directly from our definition: 
  first, note that each coordinate of the query point is chosen independently of other coordinates,
  and second, $\bv_i$ is sampled by placing a random point inside $T_i$ which by construction implies
  the depth $\bell_i$ of $\bv_i$ is a uniform random integer in $[h]$.
  (ii) This directly follows from (i): with probability $j_i/h$, the random variable $\bell_i$ is larger than
  $h-j_i$ which implies we fail Check I.
  (iii) We need to make two observations: one is that $j_i + \bell_i \ge 2$ at all times since
  $j_i \ge 1$ and second that 
  at any depth of $T_i$, except for the top level (i.e., the root), exactly half the nodes have
  a right sibling. 
  (iv) This is simply a consequence of parts (i-iii).
\end{proof}

\subparagraph{Partial Queries.}
Observe that w.l.o.g., we can assume that we first generate the dimensions $1$ to $t$ of the query, and then
the dimensions $t+1$ to $d-1$ of the query, and then the value $\by$. 
A partial query is one where only the dimensions $1$ to $t$ have been generated.
This is equivalent to only sampling $t$ random points $(\bx_i, \bz_i)$ for $1 \le i \le t$.
To be more specific, assume we have set $\bv_i = v_i$, for $1 \le i \le t$ where
each $v_i$ is a node in $T_i$. 
Then, the partial query is equivalent to the random query
$\dom_{v_1, \dots, v_t, \bv_{t+1}, \dots, \bv_{d-1}}(\bq)$ and in which the first $t$ coordinates of
$\bq$ are known (not random).  
Thus, we can still talk about the $\aj$-subproblem of a partial query; 
it could be that  the $\aj$-subproblem is already known to be undefined (this happens when
one of the nodes $u_i$, $1 \le i \le t$ is known to be undefined) but otherwise,
it is defined by defining nodes $u_1, \dots, u_t$ and the random variables
$\bu_{t+1}, \dots, \bu_{d-1}$; these latter random variables could later turn out to be undefined and 
thus rendering the $\aj$-subproblem of the query undefined. 

After sampling a partial query, we can then talk about \mdef{eligible} sums:
a sum $s$ is eligible if it could potentially be used to answer the $\aj$-subproblem once the
full query has been generated. 
Note that the emphasis is on answering the $\aj$-subproblem.
This means, there are multiple ways for a sum to be ineligible:
if $\aj$-subproblem is already known to be undefined then there are no eligible sums.
Otherwise, the defining nodes $u_1, \cdots, u_t$ are well-defined.
In this case, if
it is already known that $s$ is outside the query, or it is already known that $s$
cannot cover any points from the $\aj$-subproblem then $s$ becomes ineligible.
Final and the most important case of ineligibility is when 
$s$  is placed at a node $w_i$ which is a descendant of node $u_i \in T_i$ for some $1 \le i \le t$.
If this happens, even though $s$ can be potentially used to answer the $\aj$-subproblem, it can do so
from a different equivalent class, as the reader should remember that  we only consider sums that are stored
in the path that connects $u_i$ to $v_i$ for $1 \le i \le t$.
If a sum passes all these, then it is eligible.
Clearly, once the final query is generated, the set $S'_\aj$ is going to be a subset of the eligible sums.

\begin{definition}
    Given a partial query $\dom_{v_1, \dots, v_t, \bv_{t+1}, \dots, \bv_{d-1}}(\bq)$,
    and considering a fixed $\aj$-subproblem, 
    we define the potential function $\Phi_{v_1, \dots, v_t}$ to be the number of 
    eligible sums.
\end{definition}

\begin{restatable}{mlemma}{lemphi}\label{lem:phi}
  We have
  \[
      \E(\Phi_{\bv_1, \dots, \bv_t}\cdot \prod_{i=1}^{t}\frac{h2^{\bell_i +j_i}}{j_i} ) \le O(\pS(\D)).
  \]
\end{restatable}

To prove the above lemma, we need the following definitions and observations. 
\begin{definition}
  Consider the $\aj$-subproblem for a partial query 
  $\dom_{v_1, \dots, v_t, \bv_{t+1}, \dots, \bv_{d-1}}(\bq)$ together with corresponding nodes
  $u_1, \dots, u_t$. 
  In the representative diagram $\Gamma_i$, the \mdef{Type I region} of $u_i$ is defined as 
  a rectangular region whose bottom and left boundary are the same
  the bottom and the left boundary of $\gamma(u_i)$, its right boundary is the right boundary of
  $\gamma(u'_i)$, and its top boundary is the top boundary of $\gamma(v)$.
  We denote this region by $\tI(u_i)$. See Figure~\ref{fig:typeI}.
\end{definition}
\begin{figure}[h]
    \centering
    \includegraphics[scale=0.75]{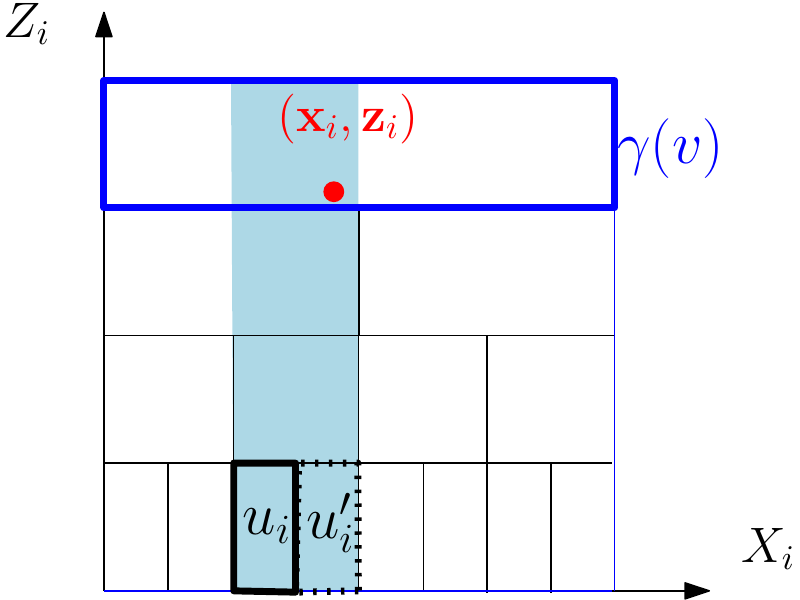}
    \caption{The type I region of $u_i$ is highlited.}
    \label{fig:typeI}
\end{figure}
\begin{restatable}{observation}{obcase}\label{ob:cases}
  Consider a partial query 
  $\dom_{v_1, \dots, v_t, \bv_{t+1}, \dots, \bv_{d-1}}(\bq)$ and assume the nodes 
  $u_1, \dots, u_t$ that correspond to the $\aj$-subproblem of the query exist. 
  A necessary condition for a sum $s$  to be eligible is that 
  $\dts_i(s)$ must lie inside $\tI(u_i)$ for $1\le i \le t$.
\end{restatable}
\begin{proof}
  As $u_i$ is a defined node in $T_i$, it means that
  we can identify the node $u'_i$, the sibling of $u_i$, and the
  node $v_i$, the node at depth $\ell_i'-j_i$ that is the ancestor of $u_i$. 
  If $\dts_i(s)$ is not inside $\tI(u_i)$, then we have a few cases:
  \begin{itemize}
    \item $\dts_i(s)$ is to the left of the left boundary of $\tI(u_i)$: 
      well in this case, $s$ cannot contain
      any point from the points in the subtree of $u_i$ so clearly it cannot be used to answer the $\aj$-subproblem.
    \item $\dts_i(s)$ is to the right of the right boundary of $\tI(u_i)$: 
      Observe that $u_i$ satisfies Check II, which means the $i$-th coordinates of the query
      is within the $i$-th coordinates of $u'_i$. Thus, in this case, it follows that
      $\dts_i(s)$ is outside the query region and so cannot be used to answer the query.
    \item $\dts_i(s)$ is below the lower boundary of $\tI(u_i)$: 
      This violates the assumption that $s$ is an eligible sum. In particular, this implies that $s$ is stored
      at the subtree of $u_i$ in $T_i$. 
    \item $\dts_i(s)$ is above the top boundary of $\tI(u_i)$: 
      This violates Observation~\ref{ob:highsubtree} as it implies $s$ is stored at a node which is
      not in the subtree of  $v_i$. 
  \end{itemize}
\end{proof}

We now return to the proof of the Lemma~\ref{lem:phi}.
\lemphi*
\begin{proof}
    If the query does not have a $\aj$-subproblem then the potential is zero and thus there is nothing
    left to prove.
    So in the rest of the proof, we will  assume $\aj$-subproblem is defined.

  Consider an eligible sum $s$ and assume $s$ has been placed at nodes $w_i$ of $T_i$, for $1\le i \le t$. 
  Let $a_i$ be the depth of $w_i$. 
  We now focus on the distribution of the random variables $\bu_i$, instead of $\bv_i$ 
  using Observation~\ref{ob:cases}:
  for $s$ to be eligible, it is necessary that 
  $\bu_i$ is selected to be a node $u_i$ with depth $\ell'$ such that $a_i \le \ell' \le a_i + j_i$ as otherwise, 
  $\dts_i(s)$ will either be below or above $\tI(u_i)$. 
  Furthermore, by Observation~\ref{ob:cases}, it follows that for every depth $\ell'_i$ such that
  $a_i \le \ell'_i \le a_i + j_i$, there exists exactly one node $u_i$ of depth $\ell'_i$ 
  for which it holds that $\dts_i(s)$ is inside $\tI(u_i)$. 
  Consider $t$ nodes $u_i$, $1 \le i \le t$ such that
  $a_i \le \ell'_i \le a_i + j_i$.
  By Observation~\ref{ob:distu}, the probability that $\bu_i = u_i$ for every $1 \le i \le t$ is 
  $O(\prod_{i=1}^t\frac{1}{h 2^{\ell'_i}})$. 
  Note that the event $\bu_i = u_i$ also uniquely determines the nodes $v_i$, $1 \le i \le t$.
  Furthermore, in this case, the depth of the node $v_i$ is $\ell_i = \ell'_i-j_i$ which means the
  contribution of $s$ to the expected value claimed in the lemma is
  \[
      O(\prod_{i=1}^t\frac{1}{h 2^{\ell'_i}})\cdot \prod_{i=1}^t\frac{h2^{\ell_i +j_i}}{j_i} = O( \prod_{i=1}^t \frac{1}{j_i}).
  \]
  Summing this over all the choices of $a_i \le \ell'_i \le a_i + j_i$ yields that the contribution of
  $s$ to the expected value is $O(1)$.
  Summing this over all sums $s$ yields the lemma. 
\end{proof}

By the above lemma, we except only few eligible sums for a random partial query. 
Let $\bad_1$ be the ``bad'' event that the nodes $\bv_1, \dots, \bv_t$ are sampled to be nodes
$v_1, \dots, v_t$ such that 
$\Phi_{v_1, \dots, v_t}\cdot \prod_{i=1}^{t}\frac{h2^{\ell_i +j_i}}{j_i}  >  |\pS(\D)/\varepsilon|$.
By Markov's inequality and Lemma~\ref{lem:phi}, $\Pr[\bad_1] = O(\varepsilon)$.

Now, fix $\bv_1 = v_1, \dots, \bv_t = v_t$.
In the rest of the proof we will assume these values are fixed and we are going to generate the rest of the query. 
Next, we  define another potential function. 

\begin{definition}
  The potential $\Psi_{w_{t+1}, \dots, w_{d-1}}$ for
    $w_{t+1} \in T_{t+1}, \dots, w_{d-1}\in T_{d-1}$, where the depth of $w_i$ in $T_i$ is $d_i$ is defined as follows.
    First define $\#_{x_{t+1}, \dots, x_{d-1}}$ for nodes $x_i \in T_i$ to  be the number of
    eligible sums $s$ such that $s$ is placed at  $x_{i}$ for $t+1 \le i \le d-1$.
    Given the nodes $w_{t+1}, \cdots, w_{d-1}$, and
    for non-negative integers $k_{t+1}, \dots, k_{d-1}$, we define $\#_{k_{t+1}, \dots, k_{d-1}}$  
    as the sum of all $\#_{w'_{t+1}, \dots, w'_{d-1}}$  over all nodes $w'_i$ where $w'_i$ has depth
    $d_i + k_i$ in $T_i$ and $w'_i$ is a descendant of $w_i$. 
    We define the potential function as follows.
\[
    \Psi_{w_{t+1}, \dots, w_{d-1}} = \sum_{k_{t+1}=0}^\infty \dots  \sum_{k_{d-1}=0}^\infty \frac{\#_{k_{t+1}, \dots, k_{d-1}}}{2^{k_{t+1}+ \dots + k_{d-1}}}.
\]
\end{definition}
\begin{restatable}{mlemma}{lempsi}\label{lem:psi}
    Having fixed the nodes $v_1, \dots, v_t$, we have,
    \[
        \E[\Psi_{\bu_{t+1}, \dots, \bu_{d-1}} \cdot \prod_{i=t+1}^{d-1}(h 2^{\bell_{i} + j_{i}})] = O(\Phi_{v_1, \dots, v_t})
    \]
    where $\bell_i$ is the depth of $\bv_i$, $\bu_i$ is the defining node of the $\aj$-subproblem, the expectation is taken over the random choices of $\bv_i$, $t+1\le i \le d-1$
    and the potential is defined to be zero if any of the nodes $\bu_i$ is undefined. 
\end{restatable}
\begin{proof}
    We consider the definition of the potential function $\Psi$.
    We observe that we can look at this potential function from a different angle.
    This potential is defined on the tuples of vertices.
    We first initialize $\Psi_{w_{t+1}, \dots, w_{d-1}}$ to $\#_{w_{t+1}, \dots, w_{d-1}}$ for every
    $w_i \in T_i$, $t+1 \le i \le d-1$.
    Then, every tuple ``dispatches'' some potential to some other tuples in the following way:
    the tuple $(w'_{t+1}, \dots, w'_{d-1})$ dispatches $\frac{\#_{w'_{t+1}, \dots, w'_{d-1}}}{2^{k_{t+1}+ \dots + k_{d-1}}}$
    potential to the tuple $(w_{t+1}, \dots, w_{d-1})$ in which $w_i$ is the ancestor of $w'_i$ in $T_i$
    that is placed $k_i$ levels higher than $w'_i$. 
    This is done for all integers $0 \le k_i$, for $t+1 \le i \le d-1$ and it is clear that by rearrnging the terms
    in the sum, it gives the same sum that was used to define the $\Psi$ potential.

    Observe that total amount of potential dispatched from a tuple $(w_{t+1}, \dots, w_{d-1})$
    is 
    \[
        \sum_{k_{t+1}=0}^\infty \dots  \sum_{k_{d-1}=0}^\infty \frac{\#_{w_{t+1}, \dots, w_{d-1}}}{2^{k_{t+1}+ \dots + k_{d-1}}} =\#_{w_{t+1}, \dots, w_{d-1}} \sum_{k_{t+1}=0}^\infty \dots  \sum_{k_{d-1}=0}^\infty \frac{1}{2^{k_{t+1}+ \dots + k_{d-1}}} = O(\#_{w_{t+1}, \dots, w_{d-1}}).
   \]
   Thus, the total amount of $\Psi$ potential is bounded by
   \[
        \sum_{w_{t+1} \in T_{t+1}} \dots \sum_{w_{d-1} \in T_{d-1}} \Psi_{w_{t+1}, \dots, w_{d-1}} = \sum_{w_{t+1} \in T_{t+1}} \dots \sum_{w_{d-1} \in T_{d-1}} O(\#_{w_{t+1}, \dots, w_{d-1}})  = O( \Phi_{v_1, \dots, v_t})
    \]
    where the last step follows from the definition of $\Phi$ potential as it counts all the eligible sums.

   Or in other words, the total amount of $\Psi$ potential is no more than the $\Phi$ potential.
   However, remember that the vertices $v_{t+1}, \dots, v_{d-1}$ are not sampled uniformly.
   Thus, to evaluate the expected value claimed in the lemma, we need to consider the exact
   distribution of the random variables $\bu_{t+1}, \dots, \bu_{d-1}$.
   We use Observation~\ref{ob:distu}.
   Define $\bell'_i = \bell_i + j_i$ and $\ell'_i = \ell_i + j_i$. 
   Thus, 
   \begin{align*}
     \E[\Psi_{\bu_{t+1}, \dots, \bu_{d-1}} \cdot (h 2^{\bell'_{t+1}}) \cdot \dots (h 2^{\bell'_{d-1}})] &= \\
     \sum_{u_{t+1} \in T_{t+1}} \dots \sum_{u_{d-1}\in T_{d-1}} \left(\Psi_{u_{t+1}, \dots, u_{d-1}} \cdot (h 2^{\ell'_{t+1}}) \cdot \dots (h 2^{\ell'_{d-1}})\right)\cdot \frac{1}{h 2^{\ell'_{t+1}}}\dots  \frac{1}{h 2^{\ell'_{d-1}}} & =  O(\Phi_{v_1, \dots, v_t}).
    \end{align*}
\end{proof}

Now we define the second bad event $\bad_2$ to be the event that
$\Psi_{u_{t+1}, \dots, u_{d-1}} \cdot (h 2^{\ell_{t+1} + j_{t+1}}) \cdot \dots (h 2^{\ell_{d-1} + j_{d-1}}) \ge  \Phi_{v_1, \dots, v_t}/\varepsilon$.
By Markov's inequality and Lemma~\ref{lem:psi},  $\Pr[\bad_2] = O(\varepsilon)$.

\subsection{Proof of the main lemma.} 
We now prove  our main lemma (Lemma~\ref{lem:mainlb} at page~\pageref{lem:mainlb}).
We restate it for convenience.
\mainlb*

Remember that we will focus on one equivalent class $\bS'_\aj$ of $\bS_\aj$.
Observe that the summation 
$\sum_{s \in \bS'_\aj} |e(s) \cap \top(\aj,\lambda)|$ 
counts how many times a point in $\top(\aj,\lambda)$ is covered by extensions of sums that cover at least $C$ points of the
$\top(\aj,\lambda)$ and this only takes into account the random choices of $\by$ as the nodes $v_1, \cdots, v_{d-1}$ have been fixed. 
As a result, $\bS'_\aj$ is a random variable that only depends on $\by$.
To make this clear, let $\S_\aj$ be the set that includes all the sums that can be part of $\bS'_\aj$ over all the
random choices of $\by$.
As a result, $\bS'_\aj$ is a random subset of $\S_\aj$. 
Observe that every sum $s \in \S_\aj$ has the property that it is stored in some node on the path from $u_i$ to $v_i$
for $1 \le i \le t$ and at the subtree of $u_i$ for $ t+1 \le i \le d-1$. 
Since $\top(\aj,\lambda)$ has exactly, $\lambda$ points, we can label them from one to $\lambda$
under some global ordering of the points (e.g., lexicographical ordering). 
Thus, let $f(g)$ be the $x$-th point in $\top(\aj,\lambda)$, $1 \le g \le \lambda$.
Also, let $m(g)$ be the number of sums $s \in \bS'_\aj$  s.t., $e(s)$ contains $f(g)$.
Then, we can do the following rewriting:
\[
  \sum_{s \in \bS'_\aj} |e(s) \cap \top(\aj,\lambda)| = \sum_{g=1}^\lambda m(g).
\] 
By linearity of expectation,
\begin{align}
  \E[\sum_{s \in \bS'_\aj} |e(s) \cap \top(\aj,\lambda)|]  = \sum_{g=1}^\lambda \E[m(g)] =\sum_{s \in \S_\aj} \sum_{g=1}^\lambda \Pr[\mbox{$s$ covers $f(g)$}, s \in \bS'_\aj] .\label{eq:target}
\end{align}

In the rest of the proof, we bound the right hand side of Eq.~\ref{eq:target} and
note that the probability is over the choices of $\by$.
Consider a particular outcome of our random trials in which the random variable
$\bv_i$ has been set to node $v_i$, for $1 \le i \le d-1$ in which 
none of the bad events $\bad_1$ and $\bad_2$ have happened. 
Set the parameter $\varepsilon$ used in the definition of these bad events to
$\varepsilon = \sqrt{\delta/\varepsilon_d}$.
Thus, none of the bad events happen with probability at least $1-O(\sqrt{\delta/\varepsilon_d})$,
conditioned on the
event that the $\aj$-subproblem of the query is defined. 
Note that can we assume the random variable $\by$ has not been assigned yet. 
This is a valid assumption since the subproblem of a query only depend on the selection of the nodes
$v_1, \dots, v_{d-1}$ and not on the $Y$-coordinate of the query. 

As  $\bad_1$ has not occurred, we have
$\Phi_{v_1, \dots, v_t}\cdot \prod_{i=1}^{t}\frac{h2^{\ell_i +j_i}}{j_i}  \le |S(\D)/\varepsilon|$.
As $\bad_2$ has not occurred either, we know that
$\Psi_{u_{t+1}, \dots, u_{d-1}} \cdot \prod_{i=t+1}^{d-1}(h 2^{\ell_{i} + j_{i}}) <  \Phi_{v_1, \dots, v_t}/\varepsilon$.
Together, they imply
\begin{align}
    \Psi_{v_{t+1}, \dots, v_{d-1}} < \frac{\Phi_{v_1, \dots, v_t}}{\varepsilon \prod_{i=t+1}^{d-1}(h 2^{\ell_{i} + j_{i}}) } \le \frac{|\pS(\D)|}{\varepsilon \prod_{i=1}^{t}\frac{h2^{\ell_i +j_i}}{j_i}}.
    \frac{1}{\varepsilon \prod_{i=t+1}^{d-1}(h 2^{\ell_{i} + j_{i}}) } = \frac{|\pS(\D)|\prod_{i=1}^{t}j_i}{\varepsilon^2 h^{d-1}\prod_{i=1}^{d-1}2^{j_i + \ell_i}}.\label{eq:psib}
\end{align}

\subparagraph{The experiment.}
To bound the sum at the Eq.~\ref{eq:target}, we will use the above inequality combined with the following experiment. 
We select a random point $\bp$ from $\top(\aj,\lambda)$ by sampling an integer $\bg \in [1, \cdots, \lambda]$ and considering
$f(\bg)$. 
We compute the probability that $f(\bg)$ can be covered by the
extension of a sum in $\bS'_{\aj}$ where the probability is computed over the choices of 
$\bg$ and the $Y$-coordinate of the query $\by$. 

We now look at the side lengths of the box $\top(\aj,\lambda)$.
The $i$-th side length of $\lambda$-top box is $\frac{1}{2^{\ell_i + j_i}}$ for
$1 \le i \le d-1$; this is because the $\aj$-subproblem was defined by nodes $u_i$ where $u_i$ 
has depth $\ell_i + j_i$.
Let $\beta$ be the side length of $\top(\aj,\lambda)$ along the $Y$-axis. 
As $\beta$ is chosen such that $\top(\aj,\lambda)$ contains
$\lambda$ points and the pointset 
well distributed, the volume of 
$\lambda$-top box  is $\Theta(\lambda/n)$.
This implies, it suffices to pick $\beta = \Theta(\frac{\lambda}{n} \prod_{i=1}^{d-1}2^{\ell_i + j_i})$.
Now remember that the $Y$-coordinate of the top boundary of the $\lambda$-top box is $y$ and the
$Y$-coordinate of its lower boundary is $y-\beta$. 

Consider a sum $s \in \S_\aj$.
Now consider the smallest box enclosing $e(s)$; w.l.o.g., we use the notation $e(s)$ to refer to this box.
For $t+1 \le i \le d-1$, the $i$-th side length of $e(s)$ is $2^{-\ell_i - j_i -\zeta_i(s)}$ because $s$ was placed at node $w_i \in T_i$
which is below $u_i$ and thus our extensions extends the $i$-dimension of the box to match that of $w_i$.
However, for $1 \le i \le t$, the $i$-th side length of $e(s)$ is $2^{-\ell_i - j_i}$.
\ignore{
Clearly, for $s$ to be in $S'_\aj$, $e(s)$ must cover $C$ points from inside $\top(\aj,\lambda)$ which requires that
$e(s)$ must intersect $\top(\aj,\gamma)$ which consequently requires that the intervals obtained by projecting $\top(\aj,\lambda)$ and $e(s)$ 
along any dimension intersect;
we divide these requirements into two groups: (a) that $Y$-coordinates of $\top(\aj,\lambda)$ and $e(s)$ intersect 
and (b) the remaining coordinates of $\top(\aj,\lambda)$ and $e(s)$ also intersect. 
For requirement (a) it is clear that we must have $y-\beta \le y_b \le y$ as otherwise we have a problem:
if $y_b < y-\beta$, then
$s$ entirely below the $\top(\aj,\gamma)$ and thus it cannot cover any points
inside $\top(\aj,\gamma)$ and if $y_b > y$, then $s$ contains at least one point from outside the query and thus
it cannot be used to cover any point in the query. 
Over random choices of the random variable $\by$, the probability of $\by-\beta \le y_b \le \by$
is at most $\beta$.
Regarding requirement (b), we look at the volumes of boxes $e(s)$ and $\top(\aj,\lambda)$.
}
We have
\begin{align}
    \vol(e(s) \cap \top(\aj,\lambda)) \le  \beta \prod_{i=1}^{t} 2^{-\ell_i - j_i } \prod_{i=t+1}^{d-1} 2^{-\ell_i - j_i -\zeta_i(s)} = \Theta(\frac{\lambda}{n} \prod_{i=t+1}^{d-1}2^{-\zeta_i(s)}).
    \label{eq:voltop}
\end{align}

Observe that we have assumed $s$ covers at least $C$ points inside $\top(\aj, \lambda)$.
However, our point set is well-distributed which implies the number of points covered by $s$ is at most
$\frac{n}{\varepsilon_d} \vol(e(s) \cap \top(\aj,\lambda))$ which by Eq.~\ref{eq:voltop}
is bounded by
$O(\frac{\lambda}{\varepsilon_d} \prod_{i=t+1}^{d-1}2^{-\zeta_i(s)})$.
We are picking the point $f(\bg)$ randomly among the $\lambda$ points inside the $\top(\aj,\lambda)$ which implies the probability that
$f(\bg)$ gets covered is at most
\begin{align}
  O(\frac{1}{\varepsilon_d} \prod_{i=t+1}^{d-1}2^{-\zeta_i(s)}).\label{eq:probc}
\end{align}
Note that above inequality is only with respect to the {\em random choices of $\bg$}
and ignores the probability of $s \in \bS'_\aj$.
However, the only necessary condition for a sum $s \in \S_\aj$ to  be in $\bS'_\aj$ is that its
$Y$-coordinate falls within the top and bottom boundaries of $\top(\aj,\lambda)$ along the $Y$-axis. 
The probability of this event is at most $\beta$ by construction.
As this probability is indepdenent of choice of $\bp$, we have 
\begin{align}
  \Pr[\mbox{$s$ covers $f(\bg)$}, s \in \bS'_\aj] = O(\frac{\beta}{\varepsilon_d} \prod_{i=t+1}^{d-1}2^{-\zeta_i(s)}).\label{eq:probcover}
\end{align}

Now we consider the definition of the potential function $\Psi$ to realize that we have
\begin{align}
    \sum_{k_{t+1}=0}^\infty \dots  \sum_{k_{d-1}=0}^\infty \frac{\#_{k_{t+1}, \dots, k_{d-1}}}{2^{k_{t+1}+ \dots + k_{d-1}}} = \Psi_{v_{t+1}, \dots, v_{d-1}} = \sum_{s\in S'_\aj} \prod_{i=t+1}^{d-1} 2^{-\zeta_i(s)}.\label{eq:psieq}
\end{align}
The left hand side is the definition of the potential function $\Psi$ where as the right hand side counts exactly the same concept:
a sum $s$ placed at depth $\ell_i + j_i + \zeta_i(s)$ of $T_i$ and at a descendant of $v_i$, for $t+1 \le i \le d-1$, contributes exactly $\prod_{i=t+1}^{d-1} 2^{-\zeta_i(s)}$
to the potential $\Psi$. 

Remember that $m(\bp)$ is the number of sums that cover a random point $\bp$ selected uniformly among the points
inside $\top(\aj,\lambda)$.
We have 
\begin{align*}
     & \sum_{s\in \S_\aj} \Pr[\mbox{$s$ covers $f(\bg)$}, s \in \bS'_\aj] = \sum_{s\in \S_\aj} O\left(\frac{\beta \prod_{i=t+1}^{d-1}2^{-\zeta_i(s)}}{\varepsilon_d}\right) =  && \eqnote{from Eq.~\ref{eq:probcover}} \\
    & O\left( \frac{\beta \Psi_{v_{t+1}, \dots, v_{d-1}}}{\varepsilon_d} \right) =  O\left( \frac{\beta }{\varepsilon_d}\cdot  \frac{|\pS(\D)|\prod_{i=1}^{t}j_i}{\varepsilon^2 h^{d-1}\prod_{i=1}^{d-1}2^{j_i + \ell_i}} \right)  = && \eqnote{from Eq.~\ref{eq:psieq} and Eq.~\ref{eq:psib}} \\
& O\left( \frac{\frac{\lambda}{n} \prod_{i=1}^{d-1}2^{\ell_i + j_i} }{\varepsilon_d}\cdot  \frac{|\pS(\D)|\prod_{i=0}^{t}j_i}{\varepsilon^2 h^{d-1}\prod_{i=1}^{d-1}2^{j_i + \ell_i}} \right) =  && \eqnote{from definition of $\beta$} \\
&O\left( \frac{\lambda}{n\varepsilon_d}\cdot  \frac{|\pS(\D)|\prod_{i=0}^{t}j_i}{\varepsilon^2 h^{d-1}} \right)  = O\left( \frac{\frac{\delta h^{d-1}}{j_1 j_2 \dots j_{d-1}}\cdot \frac{n}{\pS(\A)}}{n\varepsilon_d}\cdot  \frac{|\pS(\D)|\prod_{i=0}^{t}j_i}{\varepsilon^2 h^{d-1}} \right) =  && \eqnote{from the definition of $\lambda$} \\
& O\left( \frac{\delta}{\varepsilon_d \varepsilon^2}\right) < \frac{1}{2^d C'}. && \eqnote{from simplification and picking $\delta=O(\varepsilon_d \varepsilon^2C'^{-1}2^{-d})$ small enough}
\end{align*}
Observe that 
$\Pr[\mbox{$s$ covers $f(\bg)$}, s \in \bS'_\aj] = \frac{1}{\lambda}\sum_{g=1}^\lambda \Pr[\mbox{$s$ covers $f(g)$}, s \in \bS'_\aj]$.
Now our Main Lemma follows from plugging this  in Eq.~\ref{eq:target}.
\subsection{The Lower Bound Proof}\label{sec:lbproof}

Our proof strategy is to use Lemma~\ref{lem:mainlb} to show that the query algorithm
is forced to use a lot of sums that only cover a constant number of points inside the query,
leading to a large query time. 
\begin{theorem}
    Let $P$ be a well-distributed point set containing $\Theta(n)$ points
    in $\R^d$.
  Answering semigroup queries on $P$ using $\pS(n)$ storage and with
  $Q(n)$ query bound requires that 
  $\pS(n)\cdot Q(n) = \Omega(n (\log n \log\log n)^{d-1})$.
\end{theorem}

We pick a random query according to the distribution defined in the previous subsection.
By Lemma~\ref{lem:mainlb}, every $\aj$-subproblem for $1 \le j_i \le h/2$, has a constant probability of
being well-defined.
Let $\W$ be the set of all the well-defined subproblems. 
For a $\aj$-subproblem, let $\lambda_\aj$ be the value $\lambda$ as it is defined in Lemma~\ref{lem:mainlb}.
Observe that if a $\aj$-subproblem for $\aj = (j_1, \cdots, j_{d-1})$, is well-defined, then $\top(\aj,\lambda_\aj)$ contains
$\lambda_\aj = \frac{\delta h^{d-1}}{j_1 j_2 \dots j_{d-1}}\cdot \frac{n}{\pS(\A)}$ points. 
However, if $\aj$-subproblem or $\top(\aj,\lambda_\aj)$ is not well-defined, then 
we consider $\top(\aj,\lambda_\aj)$ to contain $0$ points.
We define the top of the query, $\top(q)$, to be the set of points
$\cup_{\aj=(j_1, \ldots, j_{d-1}), 1 \le j_1, \ldots, j_{d-1} \le h/2} \top(\aj,\lambda_\aj)$.
As each $\aj$-subproblem and $\top(\aj,\lambda_\aj)$, for $j_i \le h/2$ has a constant probability of being well-defined, we have
\begin{align}
    \E[|\top(q)|] = \sum_{\aj=(j_1, \ldots, j_{d-1}), 1 \le j_1, \ldots, j_{d-1} \le \frac{h}2} \E[\top(\aj,\lambda_\aj)] &= \Theta(1) \sum_{j_1=1}^{h/2}\ldots \sum_{j_{d-1}=1}^{h/2} \frac{\delta h^{d-1}}{j_1 j_2 \dots j_{d-1}}\cdot \frac{n}{\pS(\A)} = \nonumber \\
 &\sum_{j_1=1}^{h/2}\ldots \sum_{j_{d-2}=1}^{h/2} \frac{\delta \Theta(\log h)  h^{d-1}}{j_1 j_2 \dots j_{d-2}}\cdot \frac{n}{\pS(\A)} = \ldots =  \nonumber \\
 & \Theta\left(\frac{\delta \log^{d-1}h h^{d-1} n}{\pS(\A)}\right)\label{eq:final}.
\end{align}

\subparagraph{``Gluing'' subproblems.} 
We  now show that we can find a subset of the points $\top(q)$ that contain at least a constant fraction its points,
s.t., every sum can cover at most a constant number of points in this subset.
As a result, the total number of sums required to cover the points in $\top(q)$ is asymptotically the same 
as Eq.~\ref{eq:final}, our claimed lower bound.
The main idea is the following.
We say $s$ covers the points from a $\aj$-subproblem expensively, if $s$ covers less than $C$ points from 
$\top(\aj, \lambda_\aj)$ (otherwise, it covers them cheaply). 
If $s$ can only be used to cover points expensively from a constant
number of subproblems, then we are good.
Otherwise, we show that the number of points $s$ covers expensively is less than the number 
points $s$ covers cheaply.
But from Lemma~\ref{lem:mainlb}, we know that only a very small fraction of the points in $\top(q)$
can be covered cheaply, even when counting with multiplicity. 
As a result, most sums are expensive and cover points from a constant number of subproblem, i.e., cover a constant number of points. 
Thus, the bound of Eq.~\ref{eq:final} emerges
as an asymptotic lower bound for the query time. 

Remember that we have bounded in Eq.~\ref{eq:final} that
\begin{align}
    \E[|\top(q)|] =\Theta\left(\frac{\delta \log^{d-1}h h^{d-1} n}{\pS(\A)}\right).
\end{align}

We now show that this is an asymptotic lower bound on the query time. 
For a point $p$, in $\top(\aj,\lambda)$, if there exists a sum $s$ that covers $p$ together with at least $C$ other points
from $\top(\aj,\lambda)$, we say $p$ is \mdef{cheaply covered}. 
We denote by $\eC(q)$, the total number of times the points in $\top(q)$ are cheaply covered (this is counted with
multiplicity, i.e., if a point is covered cheaply by multiple sum, it is counted multiple times). 
By Lemma~\ref{lem:mainlb}, and using linearity of expectation we have
\begin{align}
  \E[\eC(q)] \le \frac{|\top(q)|}{C'}.\label{eq:cheap}
\end{align}

Since, by Lemma~\ref{lem:mainlb}, for every well-defined $\top(\aj,\lambda_\aj)$, on average,
only a constant fraction of the points can be cheaply covered, meaning, most sums that cover points from a 
subproblem, will be \mdef{expensive} and subsequently, 
on average a well-defined $\top(\aj,\lambda_\aj)$ will require $\Omega(|\top(\aj,\lambda_\aj)|/C)$ distinct sums
to be covered. 
If we can add these numbers together, we will obtain our lower bound.
However, there is one technical difficulty and that is the same sum $s$ can be an expensive sum but with respect to many different
subproblems (making it economical for the data structure to use).
We would like to show that there cannot be too many sums like this. 

\subparagraph{The details of gluing subproblems.}
The main idea is that if a sum $s$ is expensive with respect to a lot of subproblems, then $s$ covers a lot of points cheaply.
However, as we have a limit on how many points can be cheaply covered, this implies that a sum $s$ cannot be expensive with respect to
a lot of subproblems. 
To be able to do this, we need to understand how different subproblems are related to each other;
so far, we have treated each subproblem individually but now we have to ``glue'' them together!

Consider a query $\dom_{v_1, \ldots, v_{d-1}}(q)$ and a sum $s$ that can be used to answer this query. 
By Observation~\ref{ob:highsubtree}, $s$ is stored at the subtree of a node $v_i$ in $T_i$ for every $1 \le i \le d-1$
(see Figure~\ref{fig:intd} for an example).
Consider the point $(x_i, z_i)$ that is used to define the query. 
We know the following: $v_i$ is the unique node in $T_i$ such that $\gamma(v_i)$ contains $(x_i,z_i)$. 
Let $w_i$ be the leaf node in $T_i$ such that $r(w_i)$ contains the point $q$ and let $\pi_i$ be the path
that connects $v_i$ to  $w_i$ (shown in red in Figure~\ref{fig:intd}).
Any node $u$ in $T_i$ that hangs to the left of the path $\pi$ could be a defining node of a subproblem of
the query. 
In other words, if $u_i$ is a node $T_i$ that has a right sibling on $\pi_i$, 
then $u_i$ could be among the defining nodes of some subproblem of the query. 
Let $u^{(1)}_i, \cdots, u^{(f_i)}_i$ be the list of nodes with this property
and let $\ell_i + j^{(k)}_i$ be the depth of $u^{(k)}_i$. 
Observe that the Cartesian product
$\left\{ u^{(1)}_1, \cdots, u^{(f_1)}_1 \right\} \times \cdots \times \left\{ u^{(1)}_{d-1}, \cdots, u^{(f_{d-1})}_{d-1} \right\}$
captures all the possible tuples of $d-1$ nodes, one from each $T_i$, that are defining nodes of some subproblem of the query;
every choice in this Cartesian product will yield a subproblem and for every subproblem its tuple of $d-1$ defining nodes
can be found in this Cartesian product. 
Thus, for every $\aj \in \left\{ j^{(1)}_1, \cdots, j^{(f_1)}_1 \right\} \times \cdots \times \left\{ j^{(1)}_{d-1}, \cdots, j^{(f_{d-1})}_{d-1} \right\}$
we have a $\aj$-subproblem of the query with the corresponding defining nodes from the aforementioned Cartesian product. 

Now let us look back at the sum $s$. 
The line segment $t_i(s)$ denotes the  $X_i$-range of the sum $s$.
Let us examine its projection on the $X_iZ_i$ plane and in the representative diagram $\Gamma_i$. 
The $X_i$-range of $t_i(s)$ could be disjoint from the $X_i$-range of some prefix of the list of nodes
$u^{(1)}_i, \cdots, u^{(f_i)}_i$ as well as some suffices of this list. 
This means, there will be indices $f_i(s)$ and $e_i(s)$ such that $s$ cannot be used for any subproblem
involving nodes 
$u^{(1)}_i, \cdots, u^{(f_i(s)-1)}$ or the nodes 
$u^{(e_i(s)+1)}_i, \cdots, u^{(f_i)}$.
However, for any subproblem
$\aj \in \ej_s = \left\{ j^{(f_1(s))}_1, \cdots, j^{(e_1(s))}_1 \right\} \times \cdots \times \left\{ j^{(f_{d-1})}(s)_{d-1}, \cdots, j^{(e_{d-1}(s))}_{d-1} \right\}$ 
$s$ can potentially be used for $\aj$-subproblem, provided its $Y$-coordinate is below that of the query. 
We now estimate the volume of the intersection of $e(s)$ with $\top(\aj,\lambda_\aj)$ 
for different $\aj \in \ej_s$.

Observe that $e(s)$ and $\top(\aj,\lambda_\aj)$ will fully intersect along any dimension other than $Y$ for any
$\aj \in \ej_s$.
In fact, this property is the entire reason why we had to deal with extensions of sums rather than the sums
themselves.
However, observe that the $e(s)$ does not have a bottom boundary (or a lower bound) along the $Y$-axis and its
top boundary is fixed. 
On the other hand, the top boundary of all boxes $\top(\aj,\lambda_\aj)$ is $y$ but their bottom boundary is
variable; it is $y-\beta_\aj$ for a parameter $\beta_\aj$ that depends on the subproblem. 
Consider two subproblems, $\aj = (j_1, \cdots, j_{d-1})$ and $\aj' = (j_1+x_1, \cdots, j_{d-1} + x_{d-1})$ where
$x_i \ge 0$.
Remember that $\beta_\aj$ was defined $\beta_\aj = \Theta(\frac{\lambda_\aj}{n} \prod_{i=1}^{d-1}2^{\ell_i + j_i})$.
We now calculate the ratio $\beta_{\aj'}/\beta_\aj$ and observe that
\[
    \frac{\beta_{\aj'}}{\beta_\aj} = \Omega\left(\frac{\lambda_{\aj'}}{\lambda_\aj}\right) \prod_{i=1}^{d-1}2^{x_i} = \Omega\left(\prod_{i=1}^{d-1}\frac{j_i2^{x_i}}{j_i + x_i}\right) = \Omega(\prod_{i=1}^{d-1}2^{x_i/2})
\]
Let $y_b$ be the $Y$-coordinate of top boundary of $e(s)$.
If $e(s)$ and $\top(\aj, \lambda_\aj)$ intersect, it follows that $y-\beta_\aj \le y_b \le y$. 
As discussed, $\beta_{\aj'}$ will be larger (by a $2^{x_1/2 + \ldots + x_{d-1}/2}$ factor at least) 
which implies not only $e(s)$ and $\top(\aj',\lambda_{\aj'})$ intersect, but the volume of their intersection is
a factor $1- O(2^{-x_1/2 - \ldots - x_{d-1}/2})$ fraction of the entire volume of 
$\top(\aj',\lambda_{\aj'})$!
As a result, this means that $e(s)$ will cover almost all the points of 
$\top(\aj',\lambda_{\aj'})$ as long as $x_1 + \cdots + x_{d-1}\ge c$ for a large enough constant $c$.

Fix a value $i$, $1 \le i \le d-1$. 
Consider the $i$-th coordinate of all the subproblems $\ej_s$.
By what we have discussed, this coordinate can take any of the values
in $\left\{ j^{(f_i(s))}_i,j^{(f_i(s)+1)}_i, \cdots, j^{(e_i(s))}_i \right\}$.
Consider a sum $s$ that can be used to answer a $\aj$-subproblem for
$\aj = (j_1, \cdots, j_{d-1})$.
We consider two cases:
\begin{enumerate}
  \item For all $i$, $1 \le i \le d-1$, $j_i$ is among the $c$ largest values of the $i$-coordinate.
    It follows that there can be at most $c^{d-1} = O(1)$ such subproblems $\aj$.
  \item At least one value $j_i$ is not among the $c$ largest values of the $i$-coordinate, meaning,
    $j_i = j_i^{(k)}$ where $k < e_i(s) - c$.
    Consider value $\aj' = (j_1, \cdots, j_{i-1}, j_i^{(k+c)}, j_{i+1}, \cdots, j_{d-1}) \in \ej_s$ where we have
    only replaced the $i$-coordinate of $\aj$ with a different value. 
    And the value we have replaced it with has a rank $c$ higher.
    In this case, we know that $s$ almost entirely covers $\top(\aj',\lambda_{\aj'})$.
    Now, we can charge
    any point that $s$ covers expensively in $\aj$-subproblem to one point that $s$ covers cheaply
    in $\aj'$-subproblem. 
    It is clear that any point in $\aj'$-subproblem can be charged at most $(d-1)$ times, since they can only
    be charged once along any dimension. 
\end{enumerate}

Now we are almost done. 
If a sum $s$ can cover points from many different subproblems, then it also covers a lot of points cheaply.
However, we know that only a small fraction of the points in $\top(q)$ can be covered cheaply.
As a result, at least a constant fraction of the points in $\top(q)$ should be covered by sums that are only used
for a constant number of subproblems. 
Each such sum covers a constant number of points and thus the number of sums required to cover the points in 
$\top(q)$ is asymptotically bounded by Eq.~\ref{eq:final}. 
This concludes the proof.

\begin{figure}[h]
    \centering
    \includegraphics[scale=0.75]{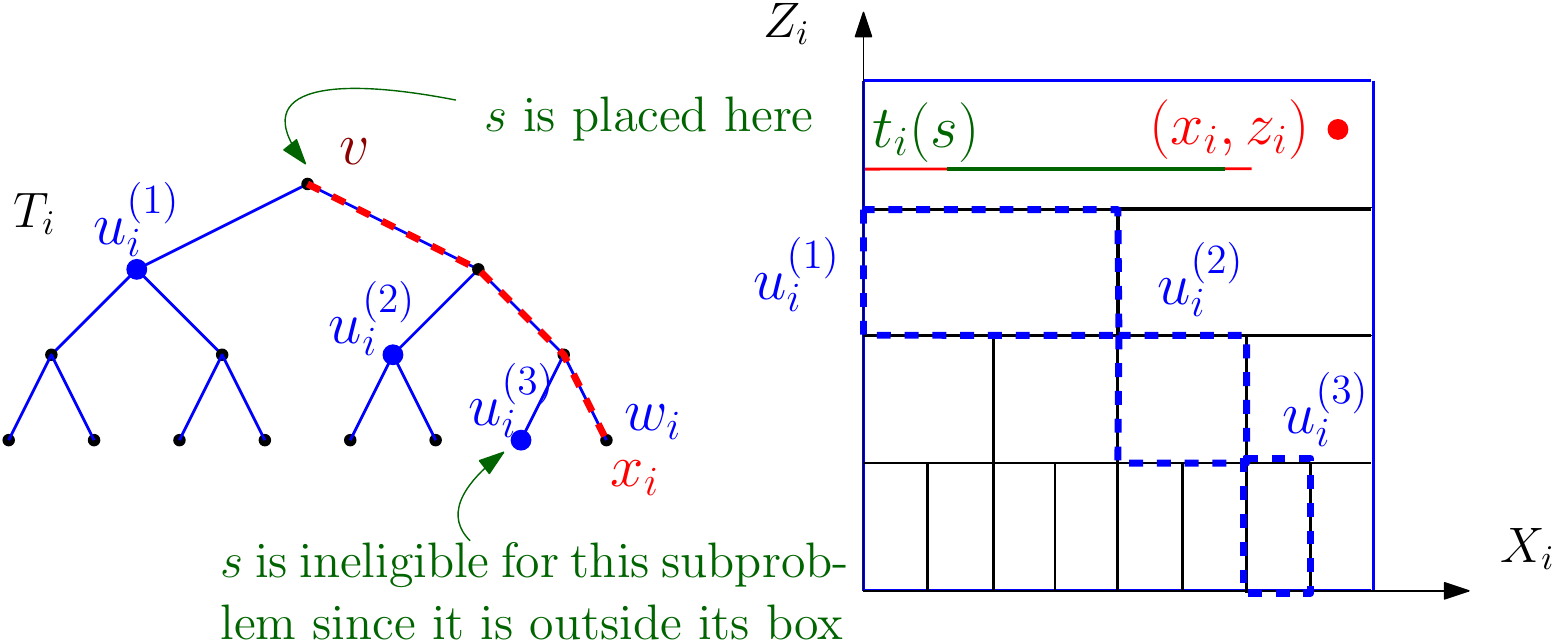}
    \caption{The query in the picture has three possible defining nodes $u_i^{(1)}$, $u_i^{(2)}$ and $u_i^{(3)}$ for its subproblems
    in $T_i$. $s$ will be ineligible for any subproblem that involves $u_i^{(3)}$ as its defining node.}
    \label{fig:intd}
\end{figure}

\section{The Upper Bounds}\label{sec:ub}
We build data structures for idempotent semigroups and for
well-distributed point sets (or a set
of $n$ points placed uniformly at random inside a square) and show 
that our analysis in the previous section is tight. 
Due to lack of space, the technical parts of the proof have been moved to the appendix but
the main idea is to simulate the phenomenon we have captured in our lower bound:
the idea that one can store sums such that the sums from different subproblems ``help'' each other.
To do that, we define the notion of ``collectively well-distributed'' point sets. 
Intuitively, collectively well-distributed point sets is a collection of point sets $\P$ where
each element of $\P$ is a well-distributed point set but importantly, certain unions of the point sets in $\P$ are also well-distributed point sets.
See Fig.~\ref{fig:col} for an example.

\begin{figure}[h]
  \centering
  \includegraphics[scale=0.5]{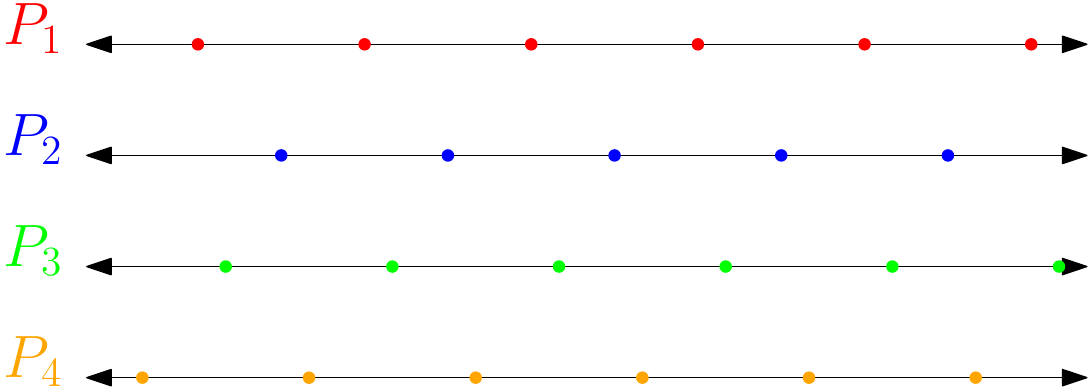}
  \caption{The point sets $P_1, P_2, P_3, P_4$ are well-distributed. For any continuous set of integers $I \subset \left\{ 1, \dots, 4 \right\}$, 
  $\cup_{i\in I}P_i$ is also well-distributed but $P_1 \cup P_3$ might not be well-distributed. }
  \label{fig:col}
\end{figure}
\begin{definition}
  Let $\sI= [t]^k$ be a set of indices, for an integer $t$ and a constant integer $k$.
  Let $\P$ be a collection of  point sets of roughly equal size indexed by $I$. 
  That is, for each $\si \in \sI$, there exists a  point set $P_\si \in \P$ containing $\Theta(n)$ points in $\R^d$.
  We say $\P$ is \mdef{collectively well-distributed} if the following holds for any $2k$ integers
  $1 \le i_1 \le j_1 \le t$, $1 \le i_2 \le j_2 \le t, \dots, $, and $1 \le i_k \le j_k \le t$:
  The point set $\cup_{i_1 \le \ell_1 \le j_1} \dots \cup_{i_k \le \ell_k \le j_k} P_{(\ell_1, \dots, \ell_k)}$ is
  well-distributed. 
\end{definition}

\begin{restatable}{mlemma}{lemcwd}\label{lem:cwd}
  For every $N$, $h$, and given constants $d$, and $k$, there is a collectively well-distributed point set $\P$
  indexed by $\mathscr{I}=[h]^k$ such that each point set in $\P$ contains $\Theta(N)$ points in $\R^d$.
\end{restatable}
\begin{proof}
  Our main idea  is that we can obtain the collection $\P$ 
  by projecting a well-distrusted point set in $\R^{d+k}$ down to $\R^d$.
  By Lemma~\ref{lem:well}, there exists a point set $\P$ of size $\Theta(Nh^k)$ such that $\P$ is well-distributed in $\R^{d+k}$.
  Consider the dimensions $d+1, d+2, \dots, d+k$ of the $(d+k)$-dimensional unit cube $\Q$ and divide each side of $\Q$ along those
  dimensions into $t$ equal pieces.
  This divides $\Q$ into $h^k$ congruent subrectangles, and we 
  naturally index them with elements of $[h]^k$, to obtain $h^k$ subrectangles $\Q_\si$, for $\si \in \sI$.
  Let $P'_\si$ be the subset of $\P$ in $\Q_\si$.
  By construction, the volume of $\Q_\si$ is $h^{-k}$, which by the properties of a well-distributed point set implies
  $|P'_\si| = \Theta(N)$.
  Let $P_\si$ be the projection of $P'_\si$ onto the first $d$-dimensions. 
  We have $|P_\si| = \Theta(N)$.

  Now consider $2k$ indices  $1 \le i_1 \le j_1 \le h$, $1 \le i_2 \le j_2 \le h, \dots, $, and $1 \le i_k \le j_k \le h$
  and the point set $X = \cup_{i_1 \le \ell_1 \le j_1} \dots \cup_{i_k \le \ell_k \le j_k} P_{(\ell_1, \dots, \ell_k)}$.
  Let $\mathscr{L} = \left\{ (\ell_1, \dots, \ell_k) |i_1 \le \ell_1 \le j_1, \dots, i_k \le \ell_k \le j_k  \right\}$
  which means $X = \cup_{\si \in \mathscr{L}}P_\si$.
  Define $\Q_\mathscr{L} = \cup_{\si\in \mathscr{L}}\Q_\si$. 
  We now need to show that $X$ is well-distributed in $\R^d$.
  We show the property (iii) of a well-distributed point set, the other property follows very similarly. 
  Consider a $d$-dimensional rectangle $r$ with ($d$-dimensional) volume $v$ inside the unit cube in $\R^d$.
  Let $r'$ be the $(d+k)$-dimensional rectangle whose projection onto the first $d$ dimensions is $r$ and whose $(d+\ell)$-th 
  side is the same as the $(d+\ell)$-th side of $\Q_\mathscr{L}$.
  Thus, the length of the $(d+\ell)$-th side of $\Q_\mathscr{L}$ is $\frac{j_\ell - i_\ell}{h}$ which implies
  the ($(d+k)$-dimensional) volume of $r'$ is $v \frac{j_1 - i_1}{h}\frac{j_2 - i_2}{h}\dots \frac{j_k - i_k}{h}  $. 
  This implies $r'$ contains $\Theta((v \prod_{\ell=1}^k\frac{j_\ell- i_\ell}{h} \cdot  N)$ points
  but observe that $|X| =  N \prod_{\ell=1}^k\Theta((j_\ell- i_\ell) )$ which implies
  $r'$ contains $\Theta( v|X|)$ points. 
\end{proof}

\subparagraph{A rough sketch.}
We  first describe a rough sketch of our approach. 
Assume the input $P$ is a well-distributed point set and that we are interested in
answering queries of the form $q = [a_1, b_1] \times \dots \times [a_{k},b_{k}] \times (-\infty, b_{k+1}] \dots \times (-\infty, b_d]$.
Let $h=\log n$.
We use Lemma~\ref{lem:cwd} to create a collection $\D$ containing $h^k$ point sets,
with each point set containing $\Theta(n/\log^k n)$ points. 
These point sets are indexed by $\mathscr{I}=[h]^k$.
Then, a point $X=(x_1, \dots, x_d)$ in the set $D_\si \in \D$ for $\si=(i_1, \dots, i_k)$ is
turned into a $(d+k)$-sided box $B(X)$ in the form of
$[\ell_1, x_1]\times \dots \times [\ell_k, x_k] \times (-\infty, x_{k+1}]\times \dots \times (-\infty,x_d]$.
The index $i_j$, $1 \le j \le k$ determines how long is the $j$-th side
of the box $B(X)$, i.e., the length of the interval $[\ell_j, x_j]$.
This length will be around $1/2^{i_j}$. 
Next, we will analyse how to answer a query $q$. 
We will show this query can be reduced to answering up to $h^k$ different subproblems using a range
tree approach, e.g., answering queries $q_{j_1, j_2, \dots, j_k}$ for (possibly) all choices of
$j_1, \dots, j_k \in [h]$.
These subproblems will correspond to covering smaller and smaller regions. 
Crucially, since $\D$ was collectively well-distributed, it follows that 
the subproblems become progressively easier to answer. 
After some careful analysis, we will show that answering 
$q_{j_1, j_2, \dots, j_k}$ requires $\frac{h^k \log^{d-k}n}{\prod_{i=1}^k j_i}$
sums, asymptotically. 
Finally, we observe that the summing this bound over all choices of 
$j_1, \dots, j_k \in [h]$ yields the desired bound and thus we prove the following theorem. 

We now return to our main result of the section.
\begin{restatable}{theorem}{thmub}\label{thm:ub}
  For a set $P$ of $n$ points placed uniformly randomly inside the unit cube in $\R^d$, one can build
  a data structure that uses  $O(n)$  storage such that a $(d+k)$-sided query can be answered with the 
  expected query bound of $O(\log^{d-1}n (\log\log n)^k)$, for $1 \le k \le d-1$.

  If $P$ is well-distributed, then the query bound can be made worst-case. 
\end{restatable}

We now present the details.

Chazelle~\cite{Chazelle.LB.II} showed that for a set of randomly placed points inside the unit cube,
one can build an efficient data structure matching his lower bound. 
The same analysis can be applied to a well-distributed point set to obtain a worst-case query bound;
after all, a well-distributed point set guarantees that a
shape of volume $v$ contains $\Theta(nv)$ points whereas a randomly placed point set only guarantees it
in the expectation. 
Furthermore, the analysis can be easily generalized to obtain a trade-off curve for when less than 
$n$ storage is used by the data structure. 
Thus, we can have the following result. 
\begin{restatable}{mlemma}{lemdomds}\label{lem:domds}\cite{Chazelle.LB.II}
  Let $n$ and $m$ be parameters such that $1 < m < n/2$ and let 
  $S$ be a well-distributed points in $\Q$ containing $n/m$ points. 
  Consider an input point set $P$ containing $n$ points. 
  We can build a data structure $D$ by 
  summing the weights of all the points of $P$ dominated by a point $s \in S$.
  $D$ will use at most $m$ storage and it can answer 
  any dominance query $q$ with the
  expected query bound of $O(m\log^{d-1}n)$ if $P$ is uniformly randomly placed in $\Q$. 
  If $P$ is a well-distributed point set, then the query bound is worst-case.
\end{restatable}

\begin{proof}[Proof summary.]
  Let $t=n/m$.
  Consider a query $q$ and let $X \subset S$ be the set of points dominated by $q$.
  Let $M$ be the maxima of $X$, i.e., subset of $X$ that are not dominated by points in $X$.
  Chazelle has shown that $|M| = O(\log^{d-1}t) = O(\log^{d-1}n)$ and that
  the volume of the region $R$ that is dominated by $q$ but not by any point in $M$
  is $O(\frac{\log^{d-1}n}{t}) = O(m \log^{d-1}n /n)$.
  So, $R$ will on average contain $O(m \log^{d-1}n)$ points of $P$. To answer $q$, we cover the points in 
  $R$ with singletons and the remaining points using the points in $M$.
  If $P$ is well-distributed, $R$ will contain $O(m \log^{d-1}n)$ points in the worst-case since it can 
  be decomposed into $O(|M|)$ rectangles. 
\end{proof}

We also need the following definition and lemma.
\begin{definition}
  Let $T$ be a balanced binary tree with height $h$, built on the interval $[z_1,z_2] \subset [0,1]$ and by repeatedly
  partitioning it in half. 
  In particular, every node $u \in T$ is assigned an interval $[a(u), b(u)] \subset [z_1, z_2]$,
  the root is assigned the interval $[z_1,z_2]$,
  the left child of $v$ is assigned the interval $[a(u), \frac{a(u)+b(u)}{2}]$ and the
  right child of $v$ is assigned the interval $[\frac{a(u)+b(u)}{2}, b(u)]$.
  Two nodes $u$ and $w$ in $T$ are said to be \mdef{adjacent} if they are at the same depth and $a(u) = b(w)$ or $b(u) = a(w)$;
  if $b(u) = a(w)$ then we say $u$ it the \mdef{left neighbor} of $w$.
  A \mdef{balanced prefix cover} of a leaf $v$ is defined as follows:
  it is a sequence of $h' \le 2h$ pairs of nodes
  $(u_1, w_1), \dots, (u_{h'},w_{h'})$ such that $u_i$ is the left neighbor of $w_i$ and 
  the interval $[z_1, a(v)]$ is the disjoint union of the intervals
  $[a(u_1),b(u_1)], [a(u_2),b(u_2)], \dots, [a(u_{h'}),b(u_{h'})]$.
  Furthermore, no three nodes $u_i, u_k, u_j$ can have the same depth. 
  See Figure~\ref{fig:binary}.
\end{definition}
\begin{figure}[h]
  \centering
  \includegraphics[scale=0.75]{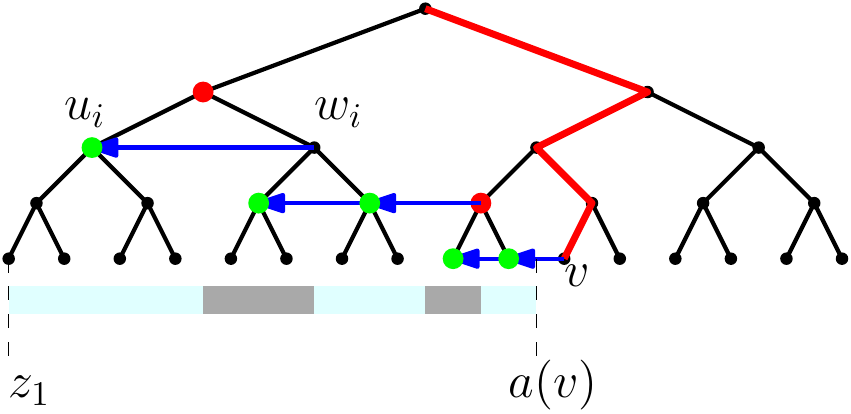}
  \caption{A balanced prefix cover. A pair $(u_i, w_i)$ is shown with a blue arrow from $w_i$ to $u_i$. 
    The intervals $[a(u_1), b(u_1)], \dots, [a(u_{h'}), b(u_{h'})]$, shown with alternating colors, disjointly cover $[z_1, a(v)]$.}
  \label{fig:binary}
\end{figure}

\begin{lemma}\label{lem:binary}
  Let $T$ be a balanced binary tree with height $h$.
  For any leaf $u$, there exists a balanced prefix cover of $u$. 
\end{lemma}
\begin{proof}
  Let $\pi$ be the path that connects the root of $T$ to $u$ and let $x_1, \dots, x_{k}$ be the nodes
  that hang to the left of $\pi$, ordered from left to right. Observe that the depth of the nodes
  $x_1, \dots, x_k$ is strictly increasing. 
  To obtain the balanced prefix cover, we use the following procedure.
  We initialize a  sequence of nodes, called the active sequence, with the list
  $x_1, \dots, x_k$ (the red nodes in Fig~\ref{fig:binary}).
  Then, we perform the following until the active sequence contains only one node. 
  Consider the first two elements of the active sequence, $x_1$ and $x_2$. 
  (i) If they have the same depth, 
  then we create the pair $(x_1, x_2)$ and then remove $x_1$  but (ii)
  otherwise, we replace $x_1$ with $x_\ell$ and $x_r$ where $x_\ell$ and $x_r$ are its left and right children.

  These operations maintain that the depth of the nodes in the active sequence is always strictly increasing, except possibly for the
  first two nodes in the sequence. 
  Furthermore, if the depth of $x_1$ is smaller than the depth of $x_2$, 
  then first we perform  operation (i) and then immediately perform operation (ii).
  The net effect is that the depth of the first element is increased by one.
  This might cause for the depth of the first element to be equal to the depth of the second element but then in the next iteration,
  operation (i) will remove the first element. 
  Thus, in overall, the operations create at most two pairs for every depth, thus, 
   the number of pairs is at most $2h$. 
\end{proof}

Our main result of the section is the following. 
\thmub*
  Let $\Q$ be the unit cube in $\R^d$.
  Assume, the queries we would like to answer 
  are in the form of $q = [a_1, b_1] \times \dots \times [a_{k},b_{k}] \times (-\infty, b_{k+1}] \dots \times (-\infty, b_d]$.
  We build a balanced binary tree for each of the first $k$ dimensions, of height $h=\log n$.
  Let $T_1, \dots, T_{k}$ be these binary trees, where $T_i$ is built on the $i$-th dimension and on the
  $i$-th side of the cube $\Q$ (the interval $[0,1]$).
  The nodes at depth $j$ of  tree $T_i$ decompose $\Q$ into $2^j$ congruent ``slabs'' using hyperplanes
  that are perpendicular to the $i$-th axis. 
  The slab of a node $v \in T_i$ is denoted by $r(v)$; this slab is defined by two hyperplanes perpendicular
  to the $i$-th axis at points $a(v)$ and $b(v)$.
  Let $m(v) = \frac{a(v)+ b(v)}2$.

  Before describing the data structure, we briefly look to see what entails to answer the query $q$. 
  Consider the $i$-th dimension of the query for $1 \le i \le k$ and the tree $T_i$. 
  Consider the highest node $v \in T_i$ such that $m(v)$ lies inside the interval $[a_i, b_i]$.
  We can now partition the interval $[a_i, b_i]$ into two intervals
  $[a_i , m(v)]$ and $[m(v), b_i]$ which in turn partitions $q$ into two queries, and over all indices $i$,
  this partitions $q$ into $2^k$ queries. 
  In the remainder of the proof, we will focus on how to answer the query that corresponds to 
  $[m(v), b_i]$, for $1 \le i \le k$, as the other queries can be handled in a similar fashion.

  We now describe the data structure. 
  We use Lemma~\ref{lem:cwd} with $h=\log n$, $N=n/h^k$ to obtain the collection of sets $\D$ containing 
  $h^k$ point sets indexed by the index set $\sI = [h]^k$.
  Using, the binary trees $T_1, \dots, T_{k}$, we turn each point in the collection of points
  $\D$ into a $(d+k)$-dimensional region which is then the data structure stores (i.e., the data structure stores the
  sum of the weights of the points inside the region).
  This is done in the following way. 
  Consider a point $X \in D_\si$ for $D_\si \in \D$ and assume $I=(i_1, \dots, i_k) \in [h]^k$ and 
  $X = (x_1, \dots, x_d)\in \R^d$.
  We turn the point $X$ into the range
  $B(X) = [\ell_1, x_1]\times [\ell_2, x_2], \dots, [\ell_k, x_k] \times (-\infty, x_{k+1}] \times \dots \times (-\infty, x_{d}]$
  where the coordinate $\ell_j$, $1 \le j \le k$, is obtained as follows:
  we look at $T_j$ and find a node $u \in T_j$ at depth $i_j$ such that $r(u)$ contains $X$; then we consider the left neighbor $v$ of $u$ and 
  we set $\ell_j$ to $a(v)$.
  In some exceptional cases $v$ might not exist, in particular, when $X$ is inside $r(u)$ and $u$ is the leftmost node at depth
  $i_j$; in such cases $B(X)$ is not defined. 
  It is clear that the data structure stores $O(n)$ sums, since the total number of points contained 
  in the point set of $\D$ is $O(n)$.

  \begin{figure}[h]
    \centering
    \includegraphics[scale=0.5]{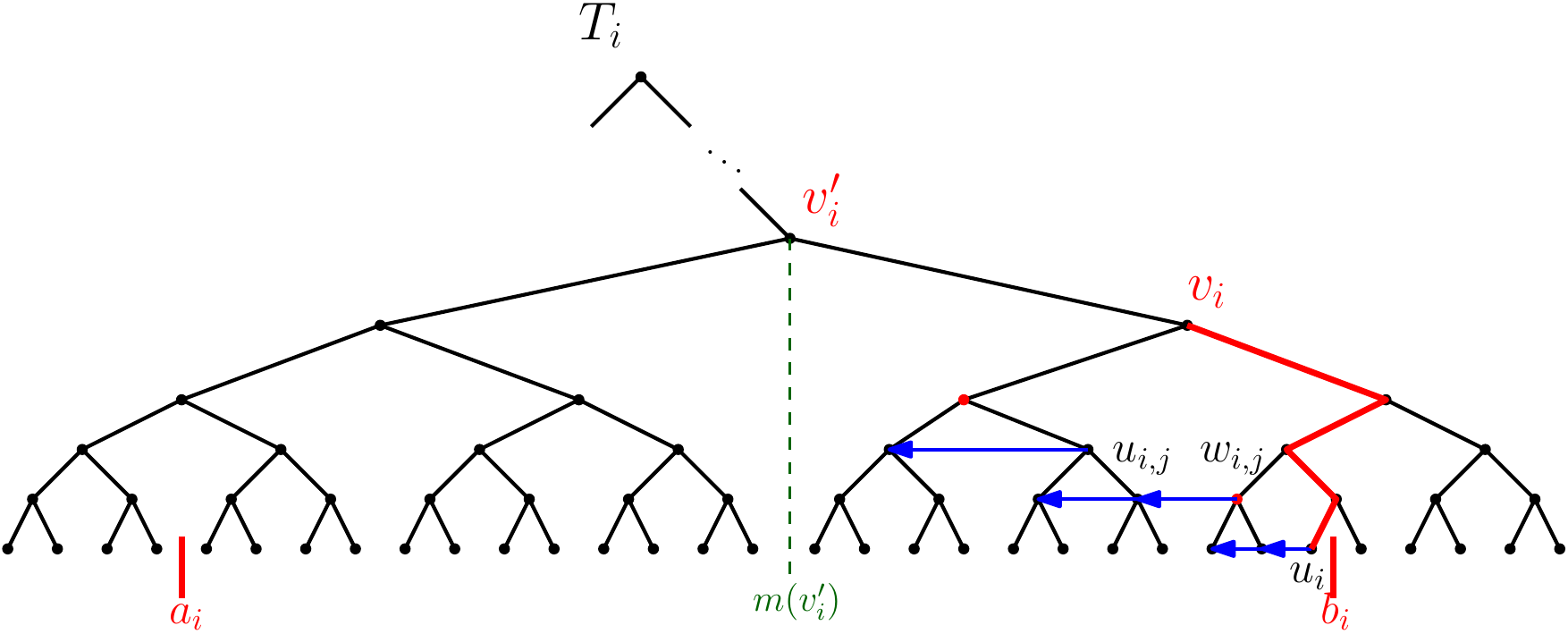}
    \caption{Answering a query.}
    \label{fig:binary-2}
  \end{figure}

  Now, consider a query range $q=[a_1, b_1] \times \dots \times [a_{k},b_{k}]
\times (-\infty, b_{k+1}] \dots \times (-\infty, b_d]$.
  Consider the interval $[a_i, b_i]$ for $1 \le i \le k$ and the tree  $T_i$.
  Let $v'_i$ be the highest node in $T_i$ such that $ a_i \le m(v'_i) \le b_i$. 
  As previously alluded, we can decompose the query into two queries at node $v'_i$:
  let $v_i$ be the right child of $v'_i$ and 
  $u_i$ be the leaf of $T_i$ such that $r(u_i)$ contains the point $(b_1, \dots, b_d)$. 
  We now consider the tree $T_i(v_i)$, i.e., the tree that hangs off at the node $v_i$. 
  By Lemma~\ref{lem:binary}, we can find a balanced prefix cover as a sequence of pairs
  $(u_{i,j}, w_{i,j})$, $1 \le j \le h'_i$  that cover the interval $[m(v'_i), a(u_i)]$.
  See Figure~\ref{fig:binary-2}.
  Thus,
  \begin{align}
    [m(v'_i), b_i] = [a(u_i), b_i] \cup \bigcup_{j=1}^{h'_i} [a(u_{i,j}),b(u_{i,j})].\label{eq:cover}
  \end{align}

  By construction, the region $[a_1, b_1] \times \dots [a_{i-1}, b_{i-1}] \times  [a(u_i), b_i] \times \ldots \times [a_{k},b_{k}]
\times (-\infty, b_{k+1}] \dots \times (-\infty, b_d]$ (where the $i$-th side of $q$ is replaced by the interval 
$[a(u_i), b_i]$) has volume at most $1/n$ and thus contains $O(1)$ input points on average;
they can be covered by singletons (or in case of a well-distributed input, $O(1)$ singletons in the
worst-case).
 To cover the rest of the query, we observe that region we would like to cover is the Cartesian product of a series
 of intervals given by Eq.~\ref{eq:cover} over all indices $1 \le i \le k$. 
 Thus, we need to cover the regions
  \begin{align}
    q_{j_1, j_2, \dots, j_k} = [a(u_{1,j_1}),b(u_{1,j_1})] \times \dots \times[a(u_{k,j_k}),b(u_{k,j_k})] \times (-\infty, b_{k+1}] \dots \times (-\infty, b_d]\label{eq:cover2}
  \end{align}
  for all choices of $1 \le j_1 \le h'_1$, and $1 \le j_2 \le h'_2$, and so on until $1 \le j_k \le h'_k$.
   \begin{figure}[h]
    \centering
    \includegraphics[scale=0.75]{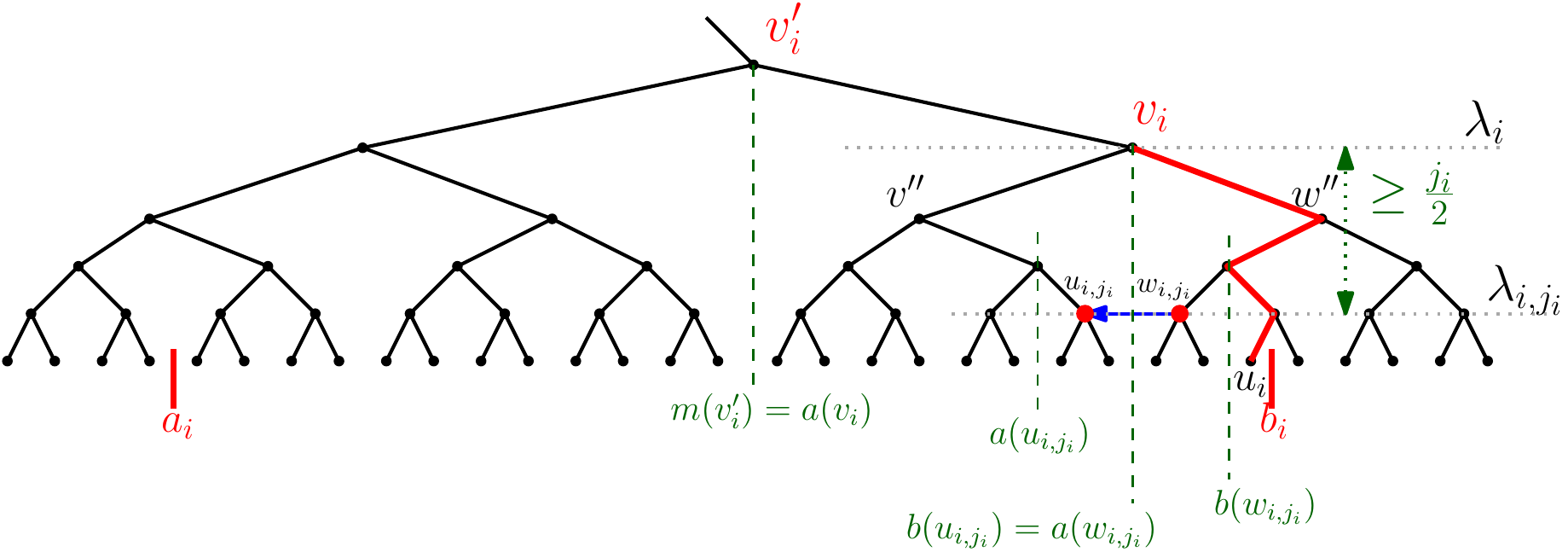}
    \caption{Answering a query.}
    \label{fig:binary-3}
  \end{figure}
 
\begin{lemma}\label{lem:crucial}
  Consider a fixed query $q_{j_1, j_2, \dots, j_k}$ obtained from query $q'$ as outlined above. 
  Consider a set $D_\si$ such that $\si=(z_2, \dots, z_k)$, $\lambda_i \le z_i \le \lambda_{i,j_i}$ for $1 \le i \le k$.
  Consider a point $X \in D_\si$ such that $X \in \cap_{i=1}^k r(w_{i,j_i})$.
  We claim the following:
  (i) $i$-th dimension, $[\ell_i, x_i]$, of the box $B(X)$ fits inside the $i$-th dimension of the query $q'$, $[a(v_i), b_i]$.
(ii) for any input point $p=(p_1, \dots, p_d) \in q_{j_1, j_2, \dots, j_k}$, the $i$-th coordinate, $p_i$, of $p$ for
$1 \le i \le k$ is within the $i$-th side of $B(X)$.
\end{lemma}
\begin{proof}
    Our first observation is the following: the node $u_{i,j_i}$ has at least $j_i/2$ ancestors between
    itself and the node $v_i$. 
    That is, if $\lambda_i$ is the depth of $v_i$ and $\lambda_{i,j_i}$ is the depth of $u_{i,j_i}$, we have
    $\lambda_{i,j_i} - \lambda_i \ge j_i/2$. 
    This observation follows because a balanced prefix cover contains at most two pairs in a given depth of a tree.
  Since $X \in \cap_{i=1}^k r(w_{i,j_i})$, it follows that $x_i \le b(w_{i,j_i}) \le b_i$
  and furthermore, since $\lambda_i \le z_i \le \lambda_{i,j_i}$, $\ell_i$ is set to
  $a(v''_i)$ for some node $v''_i$ that is a descendant of $v'_i$, which implies 
  $a(v_i) \le a(v''_i) = \ell_i$.
  This proves claim (i). 

  We now consider claim (ii). 
  Observe that $p_i$ is between $a(u_{i,j_i})$ and $b(u_{i,j_i})$ and thus
  $p_i \in [a(u_{i,j_i}), b(u_{i,j_i})]$; however, 
  $\ell_i$ is set to $a(v'')$ for a node $v''$; also $v''$'s right neighbor, $w''$, is such that
  $r(w'')$ contains $X$. 
  We have two cases: in case (i), $w''$ is also the ancestor of $u_{i,j_i}$.
  In this case $\ell_i$ is set to the same value as $a(u_{i,j_i})$ and thus the $i$-th side of the box
  $B(X)$ is the interval $[a(u_{i,j_i}), x_i]$ but since $x_i > b(u_{i,j_i})$ this interval contains
  the interval $[a(u_{i,j_i}), b(u_{i,j_i})]$. 
  In the second case, the ancestor of $u_{i,j_i}$ at depth $j_i$ is a node $v''$ that is the left neighbor of
  $w''$. 
  In this case, $\ell_i$ is set to $a(v'')$ and thus it is smaller than $a(u_{i,j_i})$ and thus once again
  the interval $[a(u_{i,j_i}), b(u_{i,j_i})]$ is contained in the $i$-th side of the box $B(X)$.  
\end{proof}

  Let $\sI'$ be the set of indices $\si$ such that 
  $\si=(z_1, \dots, z_k)$, $\lambda_i \le z_i \le \lambda_{i,j_i}$ for $1 \le i \le k$.
  Note that the previous lemma (Lemma~\ref{lem:crucial}) holds for any point $X \in D_\si$ and for any $\si \in \sI'$.
  Observe that $|\sI'| \ge \frac{j_1 j_2 \dots j_k}{2^k}$ since $\lambda_{i,j_i} - \lambda_i \ge j_i/2$.
  Let $\sX  = \cup_{\si \in \sI'}D_\si$.
  As every set in  $\D$ has asymptotically the same number of points, we have
  \[
    |\sX|  = \Omega\left( \frac{n \prod_{i=1}^k j_i}{h^k} \right).
  \]
  Consider two rectangles $R_1=\cap_{i=1}^k r(w_{i,j_i})$, and $R_2=\cap_{i=1}^k r(u_{i,j_i})$.
  Observe that they have the same $d$-dimensional volume since $u_{i,j_i}$ and $w_{i,j_i}$ 
  have equal depth. 
  Let $\alpha$ be this volume. 
  Also observe that $q_{j_1, j_2, \dots, j_k}$ is inside $R_2$. 
  Let $\sX'$ be the subset of $\sX$ that lies inside  $R_1$. 
  By Lemma~\ref{lem:crucial}, for every point $X \in \sX'$ the box $B(X)$ stored by the data structure covers
  the $i$-th dimension of every input point $p \in q_{j_1, j_2, \dots, j_k}$ for $1 \le i \le k$. 
  So in essence, we only need to take care of the last $d-k$ dimensions. 
  We do that by projection onto the last $d-k$ dimensions:
  Let $P'$, $q'$ and $\sX''$ be the projection of the subset of input points that lie inside $R_1$,
  $q_{j_1, j_2, \dots, j_k}$, and $\sX'$ onto the last $d-k$ dimensions, respectively.
  Based on what we discussed, the problem has been reduced to answering the $(d-k)$-dimensional dominance
  query $q'$ on the point set $P'$ using sums stored in $\sX''$; each such sum is the sum of the weights
  in a dominance region. 
  Now, we use Lemma~\ref{lem:domds}. 

  Since the set of input points $P$ is well-distributed, it follows that $P'$ contains $O(n\alpha)$ points.
  Since the collection $\D$ is collectively well-distributed, it follows that 
  $\sX''$ contains $\Omega(|\sX| \alpha)$ points. 
  Thus, by Lemma~\ref{lem:domds}, the $(d-k)$-dimensional query $q'$ on a set of $n' = |P'|$ points, 
  using storage $\Omega(|\sX| \alpha) = \Omega\left(\frac{n\alpha \prod_{i=1}^k j_i}{h^k} \right)$,
  can be answered with asymptotic query bound of
  \begin{align}
    \frac{h^k \log^{d-k}n}{\prod_{i=1}^k j_i}. \label{eq:fixed}
  \end{align}

  Note that the Eq.~\ref{eq:fixed} is only for a fixed query $q_{j_1, j_2, \dots, j_k}$.
  Thus, the total query bound is sum of the bound offered by Eq.~\ref{eq:fixed} over all possible choice of
  indices $j_1, \dots, j_k$:
  \begin{align*}
    \sum_{j_1=1}^{h'_1}\sum_{j_2=1}^{h'_2}\dots \sum_{j_k = 1}^{h'_k}   \frac{h^k \log^{d-k}n}{\prod_{i=1}^k j_i} = \log^dn \cdot O\left(  \sum_{j_1=1}^{h'_1}\sum_{j_2=1}^{h'_2}\dots \sum_{j_{k-1} = 1}^{h'_{k-1}}   \frac{\log\log n}{\prod_{i=1}^{k-1} j_i} \right) = O(\log^d n (\log\log n)^k).
  \end{align*}
  For a randomly placed point set, Lemma~\ref{lem:domds} offers an expected query bound and thus we obtain
  the expected query bound of $O(\log^{d-1}n(\log\log n)^k)$.
  However, if the input point set is well-distributed, we get the same bound but in the worst-case.

\section{Conclusions}\label{sec:conc}
In this paper we considered the semigroup range searching problem from a lower bound point of view.
We improved the best previous lower bound trade-off offered by Chazelle by analysing a well-distributed point set
for $(2d-1)$-sided queries for an idempotent semigroup. 
Furthermore, we showed that our analysis is tight which leads us to 
suspect that we have found an (almost) optimal lower bound for idempotent semigroups as we believe it is unlikely that a more
difficult point set exists. 
Thus, two prominent open problems emerge:
(i) Can we improve the known data structures under the {\em extra} assumption that the semigroup is idempotent?
(ii) Can we improve our lower bound under the {\em extra} assumption that the semigroup is not idempotent?
Note that the effect of idempotence on other variants of range searching was studied at least
once before~\cite{Arya.idem}.

\bibliography{ref}
\bibliographystyle{plainurl}

\end{document}